\documentclass[10pt]{amsart}
\usepackage{graphicx}
\usepackage{amscd}
\usepackage{mathrsfs}
\newtheorem{theorem}{Theorem}[section]
\newtheorem{lemma}[theorem]{Lemma}
\newtheorem{proposition}[theorem]{Proposition}
\newtheorem{corollary}[theorem]{Corollary}

\theoremstyle{definition}
\newtheorem{definition}[theorem]{Definition}

\newtheorem{remark}[theorem]{Remark}
\numberwithin{equation}{section}

\def\D{\mathscr D}
\def\F{\mathscr F}

\def\H{\mathscr H}

\def\P{\mathscr P}
\def\S{\mathscr S}
\def\A{\mathscr A}
\def\W{\mathscr W}
\def\L{\mathscr L}

\def\bkR{\mathbb R}
\def\CO{{\mathbb C}}

\newsymbol\subsetneq 2328

\begin{document}

\title[Wigner Measures in Noncommutative Quantum Mechanics]
{Wigner Measures in Noncommutative Quantum Mechanics}

\author{C. Bastos, N.C. Dias and J.N. Prata}
\address{Departamento de F\'{\i}sica and Instituto de Plasmas e Fus\~ao Nuclear, Instituto Superior T\'ecnico, Avenida Rovisco Pais 1, 1049-001 Lisboa, Portugal}
\email{cbastos@fisica.ist.utl.pt}

\address{Departamento de Matem\'atica, Universidade Lus\'ofona de Humanidades e
Tecnologias, Av. Campo Grande 376, 1749-024 Lisboa,
Portugal and Grupo de F\'{\i}sica Matem\'atica,
Universidade de Lisboa, Av. Prof. Gama Pinto 2, 1649-003,
Lisboa, Portugal}
\email{ncdias@meo.pt, joao.prata@mail.telepac.pt}



\begin{abstract}
We study the properties of quasi-distributions or Wigner measures in the context of noncommutative quantum mechanics. In particular, we obtain
necessary and sufficient conditions for a phase-space function to be a noncommutative Wigner measure, for a Gaussian to be a noncommutative
Wigner measure, and derive certain properties of the marginal distributions which are not shared by ordinary Wigner measures. Moreover, we
derive the Robertson-Schr\"odinger uncertainty principle. Finally, we show explicitly how the set of noncommutative Wigner measures relates to
the sets of Liouville and (commutative) Wigner measures.
\end{abstract}

{\maketitle }

\begin{section}{Introduction}

In this work we address several features of phase-space quasi-distributions in the context of a canonical noncommutative extension of
non-relativistic quantum mechanics. This type of system has been recently considered by various authors
\cite{Bastos,Bastos1,Bastos2,Bertolami1,Demetrian,Dias5,Gamboa,Horvathy,Nair} as a simplified model for the more elaborate noncommutative field
theories \cite{Douglas}. The motivation comes mainly from string theory \cite{Seiberg} and noncommutative geometry \cite{Connes,Madore}.
Although one may address noncommutative quantum mechanics with a standard operator formulation, we feel that the Weyl-Wigner formulation is more
adequate. Indeed, (i) it places position and momentum on equal footing, (ii) the extra noncommutativity is trivially encapsulated in a modified
Moyal product, (iii) the passage from noncommutative to ordinary quantum mechanics is more transparent, (iv) the difference between ordinary and
noncommutative quantum mechanical systems is manifest when one considers the sets of quasi-distributions of the two theories. In \cite{Bastos1}
we initiated a systematic study of the Weyl-Wigner formulation of noncommutative quantum mechanics. There, we mainly focused on the algebraic
structure of the theory. The aim of the present work is to further develop this formulation by addressing a number of issues related to the
characterization of the set of states of the theory. These will be called the noncommutative Wigner measures (NCWMs).

In connection with (iii) and (iv), some of us studied the emergence of ordinary quantum mechanics in the realm of noncommutative quantum
mechanics \cite{Dias5}. A noncommutative Brownian particle was placed in interaction with an external heat bath of harmonic oscillators. We were
looking for an estimate for the time scale of this noncommutative-commutative transition. Using a decoherence approach we were able to produce
such an estimate for this simple system. In more general cases, however, we are required to established rigorous criteria for assessing whether such
transition takes place. These criteria hinge upon the difference between the states of a system in ordinary and in noncommutative quantum
mechanics. This is analogous to the comparison between classical and quantum states. The phase space framework is more suitable to address this issue. Classical states are described by positive
Liouville measures in phase-space and they need not satisfy the Heisenberg uncertainty relations. In contrast, quantum mechanical phase-space
quasi-distributions (Wigner measures) need not be positive, but must comply with the Heisenberg constraints.

The difference between the space of states in ordinary and noncommutative quantum mechanics is also more transparent in the phase-space framework.
In the standard operator formulation of noncommutative quantum mechanics the space of states
is the usual Hilbert space of square integrable functions. From this point of view, the space of states of ordinary and noncommutative quantum
mechanics coincide and one cannot tell one from the other. In contrast with this situation, if one resorts to the Weyl-Wigner formulation, one
obtains two disparate phase-space representations: the star-product and the set of phase-space quasi-distributions are different for the two
theories. This constitutes a clear advantage if one wants to study the emergence of ordinary quantum mechanics from a more fundamental
noncommutative quantum theory \cite{Dias5}. In this respect, our work here will culminate in Lemma 4.14 and Figure 1. These reveal that, in the
wider set of phase-space real and normalizable functions, one can find the subsets of Liouville measures (the classical states), the Wigner
measures (the quantum states) and the noncommutative Wigner measures (the noncommutative quantum states), and that these have non-trivial mutual
intersections. This constitutes a natural framework to study transitions from one theory to another. A word of caution is however in order. We
will prove Lemma 4.14 for $2$-dimensional systems only. However, we believe that similar results hold for higher dimensions. The reason is that
the $d=2$ case can be obtained from the higher dimensional ones by tracing out suitable degrees of freedom. We hope to return to this issue in a
future work.

Another interesting problem that one faces when dealing with noncommutative quantum mechanics or noncommutative field theories is the fact that
the noncommutative corrections are too small to be observed experimentally. This arises from upper bounds on the noncommutative parameters such
as those found in refs. \cite{Bertolami1,Carroll}. However our results stated in Theorems 4.2 and 4.3 and in Lemmata 4.4 and 4.5 may provide a
framework to obtain physical predictions which can be tested experimentally. Indeed, in Lemma 4.4, we construct certain NCWMs which maximize a
certain functional. These measures are not Wigner measures. Moreover, they look like the tensor product of two lower dimensional states. It is
then natural to look for entangled combinations of such states as the one in (\ref{Eq4.1.40}). It is well known that entanglement "amplifies"
the quantum nature of a system, so that its quantumness can be observed macroscopically. We hope that the entangled states of the form
(\ref{Eq4.1.40}) which exist only in noncommutative quantum mechanics may play this role here. In a certain sense, the formalism seems ripe to
develop a continuous variable quantum computation and quantum information for noncommutative quantum mechanics \cite{Giedke}.

Here is a brief outlook of our work. In section 2, we recapitulate certain aspects of the Weyl-Wigner formulation which will be generalized to
noncommutative quantum mechanics in the sequel. In section 3 we address the Weyl-Wigner formulation of a $d$-dimensional noncommutative quantum
mechanical system. We define the notion of NCWM and state some of its properties in $d$ dimensions. In
particular, (i) we obtain the uncertainty principle for NCWM in the Robertson-Schr\"odinger form; (ii) we establish an upper bound on the purity
of a NCWM; (iii) we derive the Hudson Theorem which classifies the set of positive, pure state NCWMs; (iv) we obtain necessary and sufficient
conditions for a Gaussian to be a NCWM or a pure state NCWM (Littlejohn's Theorem); (v) we define the linear transformations which map a NCWM to
another NCWM and which leave the uncertainty principle unchanged. In section 4, we specialize to the case $d=2$: (i) we show that the marginal
distributions are not necessarily non-negative and that they must satisfy certain purity-type constraints which are not required in ordinary
quantum mechanics; (ii) we show that these bounds can be saturated by certain states which have the structure of a tensor product of two lower
dimensional quasi-distributions; (iii) we derive necessary and sufficient conditions for a phase-space function to be a {\it bona fide} NCWM;
(iv) we establish the relation between the states in classical mechanics, and in ordinary and noncommutative quantum mechanics (see figure1).

\section{Weyl-Wigner formulation of quantum mechanics}

In this section we review some basic facts about the Weyl-Wigner formulation of ordinary quantum mechanics which will be relevant for the sequel. Our analysis is restricted to the case of flat phase-spaces. The reader is referred to Refs. \cite{Bayen1,Bayen2,Dias1,Dias2,Dias3,Dias4,Dubin,Ellinas,Folland,Pool,Segal,Vey,Wigner} for a more detailed presentation and to Refs. \cite{Bordemann,Fedosov1,Fedosov2,Kontsevich,Wilde} for the generalization of the formalism to the non-flat case.

Let us then settle down the preliminaries: we consider a $d$-dimensional dynamical system, such that its classical formulation lives in the flat phase-space $T^*M \simeq (\bkR^d )^* \times \bkR^d \simeq \bkR^{2d}$. A global Darboux chart can be defined on $T^*M$
\begin{eqnarray}
\xi&=&(R,\Pi)=(R_i, \Pi_i),\,\, i=1, \cdots,d,\nonumber\\
\xi_{\alpha}=R_{\alpha},\,\, \alpha&=&1,..,d \hspace{0.2cm} \mbox{and} \quad \xi_{\alpha}=\Pi_{\alpha-d}, \,\,\alpha=d+1,..,2d
\label{Eq2.1}
\end{eqnarray}
in terms of which the sympletic structure reads $dR_i \wedge d\Pi_i$. In the sequel the Latin letters run from $1$ to $d$ (e.g. $i,j,k, \cdots = 1, \cdots, d$), whereas the
Greek letters stand for phase-space indices (e.g. $\alpha, \beta, \gamma, \cdots = 1, \cdots, 2d$), unless otherwise stated. Moreover, summation over repeated indices is assumed.
Upon quantization, the set $\{\hat{\xi}_{\alpha},\, \alpha = 1,..,2d\}$ satisfies the commutation relations of the standard Heisenberg algebra:
\begin{equation}
\left[\hat{\xi}_{\alpha}, \hat{\xi}_{\beta} \right] = i \hbar j_{\alpha \beta}, \hspace{1 cm}
{\bf J}=
\left(
\begin{array}{c c}
{\bf 0}_{d \times d} & {\bf I}_{d \times d}\\
- {\bf I}_{d \times d} & {\bf 0}_{d \times d}
\end{array}
\right),
\label{Eq2.2}
\end{equation}
where $j_{\alpha \beta}$ are the components of the matrix ${\bf J}$. Moreover $\{\hat R_i,\, i=1,..,d \}$ constitute a complete set of commuting observables. Let us denote by $|R\rangle$ the general eigenstate of $\hat R$ associated
to the array of eigenvalues $R_i,\, i=1,..,d$ and spanning the Hilbert space $\H =L_2(\bkR^d,dR)$ of
complex valued functions $\psi :\bkR^{d}\longrightarrow \CO$ $(\psi (R) = \langle R| \psi\rangle)$, which are square integrable with respect to
the standard Lesbegue measure $dR$.
The scalar product in $\H$ is given by:
\begin{equation}
(\psi,\phi)_{\H}=\int_{\bkR^d} \,\overline{\psi(R)} \phi(R)  dR
\label{Eq2.3}
\end{equation}
where the over-bar denotes complex conjugation.\\

\noindent We now introduce the Gel'fand triple of vector spaces
\cite{Antoine1,Antoine2,Bohm,Gelfand,Roberts}:
\begin{equation}
\S(\bkR^d) \subset \H \subset \S'(\bkR^d),
\label{Eq2.4}
\end{equation}
where $\S(\bkR^d)$ is the space of all complex valued functions $t(R)$ that are infinitely smooth and, as $||R|| \to \infty$, they and all their partial derivatives decay to zero faster than any power of $1/||R||$. $\S (\bkR^d)$ is the space of rapid descent test functions \cite{Grubb,Hormander,Zemanian}, and $\S'(\bkR^d)$ is its dual, i.e. the space of tempered distributions. In analogy with (\ref{Eq2.4}) let us also introduce the triple:
\begin{equation}
\S (\bkR^{2d}) \subset \F=L_2(\bkR^{2d},dRd\Pi) \subset \S'(\bkR^{2d}),
\label{Eq2.5}
\end{equation}
where $\F $ is the set of square integrable phase-space functions with scalar product:
\begin{equation}
(F,G)_{\F}={1\over(2\pi \hbar)^d} \int_{\bkR^{2d}}  \,\, \overline{F(\xi)} G(\xi)  d \xi.
\label{Eq2.6}
\end{equation}
Finally, let $\hat{\S}'$ be the set of linear operators admitting a representation of the form \cite{Dubin}:
\begin{equation}
\hat A: \S(\bkR^d) \longrightarrow \S'(\bkR^d); \, \, \psi(R) \longrightarrow (\hat A \psi)(R)= \int_{\bkR^d}  \, A_K(R,R') \psi(R') dR',
\label{Eq2.7}
\end{equation}
where $A_K(R,R')=\langle R|\hat A|R'\rangle \in \S'(\bkR^{2d})$ is a distributional kernel.
The elements of $\hat{\S}'$ are named generalized operators.\\

\noindent
The Weyl-Wigner transform is the linear one-to-one invertible map \cite{Bracken,Dubin,Segal}:
\begin{eqnarray}
W_{\xi}: \hat{\S}'  & \longrightarrow & \S'(\bkR^{2d});  \nonumber\\
\hat A & \longrightarrow & A(R,\Pi)= W_{\xi}(\hat A) = \hbar^d  \int_{\bkR^d} \, e^{-i\Pi\cdot y}
A_K(R+\frac{\hbar}{2} y,R-\frac{\hbar}{2} y) dy  \nonumber \\
&&\hspace{3.2cm} =\hbar^d  \int_{\bkR^d}  \, e^{-i\Pi\cdot y}\,
 \langle R+\frac{\hbar}{2} y| \hat A |R-\frac{\hbar}{2} y\rangle dy
\label{Eq2.8}
\end{eqnarray}
where the Fourier transform is taken in the usual generalized way\footnote{The Fourier transform $T_F$ of a generalized function $B \in \S'(\bkR^n)$ (for $n \ge 1$) is another generalized function $T_F[B]\in \S'(\bkR^n)$ which is defined by $ \langle T_F[B],t\rangle=\langle B,T_F[t]\rangle$ for all $t \in \S(\bkR^n)$ \cite{Grubb,Hormander,Zemanian} and where $\langle A,t\rangle$ denotes the action of a distribution $A \in \S'(\bkR^n)$ on the test function $t \in \S (\bkR^n)$.} and the second form of the Weyl-Wigner map in terms of Dirac's bra and ket notation is more standard.

There are two important restrictions of $W_{\xi}$:

(1) The first one is to the vector space $\hat{\F}$ of Hilbert-Schmidt operators on $\H$, which admit a representation of the form (\ref{Eq2.7}) with $A_K(R,R') \in \F$, regarded as
an algebra with respect to the standard operator product, which is an inner operation in $\hat{\F}$.
In this space we may also introduce the inner product $(\hat A,\hat B)_{\hat{\F}} \equiv tr(\hat A^{\dagger} \hat B)$ and the Weyl-Wigner map
$ W_{\xi}: \hat{\F} \longrightarrow \F$ becomes a one-to-one invertible unitary transformation.

(2) The second one is to the enveloping algebra $\hat{\A}({\H})$ of the Heisenberg-Weyl
Lie algebra which contains all polynomials of the fundamental operators $\hat R, \hat{\Pi}$ and $\hat{I}$
modulo the ideal generated by the Heisenberg commutation relations. In this case the Weyl-Wigner
transform $W_{\xi}:\hat{\A}(\H) \longrightarrow \A(\bkR^{2d})$ becomes a one-to-one invertible map from $\hat{\A}(\H)$ to the algebra $\A (\bkR^{2d})$ of polynomial
functions on $\bkR^{2d}$. In particular $W_{\xi}(\hat I)=1$, $W_{\xi}(\hat R)=R$ and $W_{\xi}(\hat{\Pi})=\Pi$.\\

\noindent
The previous restrictions can be promoted to isomorphisms, if $\F$ and $\A (\bkR^{2d})$
are endowed with a suitable product. This is defined by:
\begin{equation}
W_{\xi}(\hat A) \star_{\hbar} W_{\xi}(\hat B) := W_{\xi}(\hat A \hat B)
\label{Eq2.9}
\end{equation}
for $\hat A,\hat B \in \hat{\S}'$ such that $\hat A\hat B \in \hat{\S}'$.
The $\star$-product admits the kernel representation:
\begin{equation}
A(\xi) \star_{\hbar} B(\xi) = {1\over(\pi \hbar)^{2d}} \int_{\bkR^{2d}} \int_{\bkR^{2d}} ~ A(\xi') B (\xi'') \times  \exp \left[ - \frac{2i}{\hbar} \left(\xi- \xi'\right)^T {\bf J} \left(\xi'' - \xi \right) \right] d \xi'd \xi''
\label{Eq2.10}
\end{equation}
and is an inner operation in $\F$ as well as in $\A (\bkR^{2d})$. The previous formula is also valid if we want to compute $A \star_{\hbar} B$ with $A \in \F$ and $B \in \A (\bkR^{2d})$, in which case $A \star_{\hbar} B \in \S' (\bkR^{2d})$. On the other hand, if $A \in \A (\bkR^{2d})$ and $B \in \A (\bkR^{2d}) \cup \F$ the $\star$-product can also be written in the well-known form \cite{Groenewold,Moyal}:
\begin{equation}
A(\xi) \star_{\hbar} B(\xi) = A(\xi) e^{ \frac{i\hbar}{2} \buildrel{\leftarrow}\over\partial_{\xi_{\alpha}} j_{\alpha \beta} \buildrel{\rightarrow}\over\partial_{\xi_{\beta}}} B(\xi),
\label{Eq2.11}
\end{equation}
where $\partial_{\xi_{\alpha}} A(\xi) = \frac{\partial}{\partial_{\xi_{\alpha}}} A (\xi)$ and $j_{\alpha \beta}$ are the components of the sympletic matrix (\ref{Eq2.2}).

When the Weyl-Wigner map is applied to a density matrix $\hat{\rho} \in \hat{{\F}}$, we get the celebrated Wigner measure or quasi-distribution:
\begin{equation}
f^C (\xi) = \frac{1}{(2 \pi \hbar)^d} W_{\xi} (\hat{\rho}) (\xi).
\label{Eq2.12}
\end{equation}
The superscript "C" will become clear in the sequel. If the system is in a pure state $\hat{\rho} = |\psi>< \psi|$, then the Wigner measure reads:
\begin{equation}
f^C (R, \Pi) = \frac{1}{(\pi \hbar)^d} \int_{\bkR^d}~ e^{-2 i y \cdot \Pi /\hbar} \overline{\psi (R-y)} \psi (R+y)  dy .
\label{Eq2.13}
\end{equation}
Mixed states are just convex combinations of the latter. It is important to recapitulate some properties of Wigner measures. They are real and normalized phase-space functions which admit marginal ditributions. For instance, for a pure state (\ref{Eq2.13}) these read:
\begin{equation}
\begin{array}{l}
\P_R (R) = \int_{\bkR^d} ~ f^C (R, \Pi) d \Pi = | \psi (R)|^2  \ge 0, \\
 \\
 \P_{\Pi} (\Pi) = \int_{\bkR^d}  ~ f^C (R, \Pi) d R = | \hat{\psi} (\Pi)|^2 \ge 0,
\end{array}
\label{Eq2.14}
\end{equation}
where $\hat{\psi} (\Pi)$ is the Fourier transform of $\psi (R)$. However, even though their marginal distributions are {\it bona fide} probability densities for position or momentum, the Wigner measures themselves cannot be regarded as joint probability distributions for position and momentum. The reason is that they are not necessarily everywhere non-negative. This is usually regarded as a manifestation of Heisenberg's uncertainty principle. Nevertheless, the existence of strictly positive Wigner measures is not precluded. It is a difficult (and unfinished) task to classify the entire set of Wigner functions which are everywhere non-negative. An important step in that direction is Hudson's beautiful theorem \cite{Hudson,Soto}:

\begin{theorem} {\bf (Hudson, Soto, Claverie)} Let $\psi \in L^2 (\bkR^d, d R)$ be a state vector. The Wigner measure of $\psi$ is non-negative iff $\psi$ is a Gaussian state.
\end{theorem}

There is thus far no analogous result for mixed states (see \cite{Werner1} for some attempts at classifying positive Wigner measures of mixed states). A class of positive Wigner measures of mixed states can be constructed by convoluting Wigner measures with suitable kernels (see below).

Another important property of Wigner measures is the following bound:
\begin{equation}
\int_{\bkR^{2d}}   \left[ f^C (\xi) \right]^2  ~d \xi \le \frac{1}{(2 \pi \hbar)^d}.
\label{Eq2.15}
\end{equation}
The equality holds iff the state is pure. For this reason, one calls the integral on the left-hand side the {\it purity} of the system.

Equally interesting is the issue of whether the previous properties are necessary and sufficient conditions for a phase-space function to be a Wigner measure. The answer is no. In fact these are just necesssary conditions. A phase-space function $f (\xi)$ is a Wigner measure iff there exits $b(\xi) \in \F$ such that:
\begin{eqnarray}
(i) & \int_{\bkR^{2d}} | b(z)|^2 ~d \xi =1 \label{Eq2.16}\\
(ii) & f( \xi) = \overline{b (\xi)} \star_{\hbar} b (\xi) \label{Eq2.17}
\end{eqnarray}
If the state is pure, then we may take $b (\xi) = (2 \pi \hbar)^{d/2} f (\xi)$, i.e.:
\begin{equation}
f_{pure}^C (\xi) \star_{\hbar} f_{pure}^C (\xi) = \frac{1}{(2 \pi \hbar)^d} f_{pure}^C (\xi).
\label{Eq2.18}
\end{equation}
These necessary and sufficient requirements are equivalent to another set of conditions called the {\it KLM conditions}. To state the latter, we first need the concept of {\it symplectic Fourier transform}:

\begin{definition} Let $f(\xi) \in \F$. We define its symplectic Fourier transform according to:
\begin{equation}
\tilde f^{{\bf J}} (a) = \int_{\bkR^{2d}} f(\xi) \exp \left( -i a^T {\bf J} \xi \right) ~d \xi .
\label{Eq2.19}
\end{equation}
The superscript "${\bf J}$" will be useful for the sequel. The formula (\ref{Eq2.19}) can be inverted:
\begin{equation}
f (\xi) = \frac{1}{(2 \pi )^{2 d}} \int_{\bkR^{2d}}\tilde f^{{\bf J}} (a) \exp \left( i a^T {\bf J} \xi \right) ~ d a .
\label{Eq2.20}
\end{equation}
\end{definition}

\begin{definition} The symplectic Fourier transform $\tilde f^{{\bf J}} (a)$ is said to be of the $\alpha$-positive type if the $m \times m$ matrix with entries
\begin{equation}
M_{jk} = \tilde f^{{\bf J}} (a_j - a_k) \exp \left( - \frac{i \alpha}{2} a_k^T {\bf J} a_j \right)
\label{Eq2.21}
\end{equation}
is hermitian and non-negative for any positive integer $m$ and any set of $m$ points $a_1, \cdots, a_m$ in the dual of the phase-space. By abuse of language we sometimes say that $f(\xi)$ is of the $\alpha$-positive type.
\end{definition}

With these definitions one can state the KLM (Kastler, Loupias, Miracle-Sole \cite{Kastler,Loupias,Narcowich1}) conditions, equivalent to (\ref{Eq2.16},\ref{Eq2.17}):

\begin{theorem} The phase-space function $f(\xi)$ is a Wigner measure, iff its symplectic Fourier transform $\tilde f^{{\bf J}} (a)$ satisfies the KLM conditions:
\begin{eqnarray}
(i) & \tilde f^{{\bf J}} (0)  & =1 \label{Eq2.22} \\
(ii) & \tilde f^{{\bf J}} (a) & {\mbox{is continuous and of $\hbar$-positive type.}} \label{Eq2.23}
\end{eqnarray}
\end{theorem}

The concept of Narcowich-Wigner (NW) spectrum is useful in this context.

\begin{definition} The Narcowich-Wigner spectrum of a phase-space function $f(\xi) \in \F$ is the set:
\begin{equation}
\W (f) = \left\{ \alpha \in \bkR \left| \tilde f^{{\bf J}} (a) {\mbox{ is of the $\alpha$-positive type}} \right. \right\}.
\label{Eq2.24}
\end{equation}
Consequently, one may say that if $f$ is a Wigner measure, then $\hbar \in \W (f)$.
\end{definition}

If we analyze carefully the equivalence between the necessary and suffcient conditions (\ref{Eq2.16},\ref{Eq2.17}) and the set of KLM conditions (\ref{Eq2.22},\ref{Eq2.23}), we conclude that there is nothing special about $\hbar$. In fact, it is trivial to conclude that:

\begin{theorem} Let $f \in \F$, $\tilde f^{{\bf J}} $ be its symplectic Fourier transform, and $\alpha \in \bkR \backslash \left\{0 \right\}$. Then the following sets of conditions are equivalent:
\begin{equation}
\begin{array}{l l}
(i) & {\mbox{$\tilde f^{{\bf J}} (a)$ is continuous, $\tilde f^{{\bf J}} (0)=1$ and $\tilde f^{{\bf J}} (a)  $ is of the $\alpha$-positive type.}} \\
& \\
(ii) & {\mbox{There exists $b (\xi) \in \F$ such that $\int_{\bkR^{2d}} |b( \xi)|^2 ~d \xi  =1$ and $f(\xi) = \overline{b (\xi)} \star_{\alpha} b (\xi)$.}}
\end{array}
\label{Eq2.25}
\end{equation}

\end{theorem}

Here $\star_{\alpha}$ denotes the Moyal product (\ref{Eq2.10},\ref{Eq2.11}) with $\hbar$ replaced by $\alpha$. The case $\alpha =0$ is singular. Functions $\tilde f^{{\bf J}} (a)$ which are of $0$-positive type correspond, according to Bochner's theorem, to phase-space functions $f(\xi)$ which are everywhere non-negative.

From the previous theorem and the fact that $a(\xi) \star_{\alpha} b (\xi) = b(\xi) \star_{- \alpha} a (\xi)$, we can check that:

\begin{corollary} Let $f \in \F$. If $\alpha \in \W (f)$, then $- \alpha \in \W (f)$.
\end{corollary}

Let $f \natural g$ denote the convolution:
\begin{equation}
(f \natural g ) (\xi) : = \int_{\bkR^{2d}} ~ f( \xi - \xi') g (\xi') d \xi'.
\label{Eq2.25A}
\end{equation}
We use this somewhat unusual notation for the convolution to avoid confusion with the star product.

An important result concerning the convolution of phase-space functions was proved in ref.\cite{Narcowich1}, using the fact that the Schur (or Hadamard) product of hermitian and non-negative matrices is again a hermitian and non-negative matrix:

\begin{theorem} The NW spectrum of the convolution $f \natural g$ contains all elements of the form $\alpha_1 + \alpha_2$ with $\alpha_1 \in \W (f)$ and $\alpha_2 \in \W (g)$.
\end{theorem}

Nevertheless, it is not true that $\W (f \natural g) = \left\{ \alpha_1 + \alpha_2 \left| \alpha_1 \in \W (f), ~ \alpha_2 \in \W (g) \right. \right\}$. Counter-examples can be found in \cite{Werner1}.

An immediate consequence of Corollary 2.7 and Theorem 2.8 is that we can construct positive Wigner measures by convoluting a Wigner measure $f^C(\xi)$ with another Wigner measure $g^C (\xi)$, which, in addition to being of $\hbar$-positive type, is also of the $0$-positive or $2 \hbar$-positive type. In ref.\cite{Narcowich1} functions $g^C(\xi)$ with these characteristics where explicitly constructed. Moreover, by resorting to this concept of $\hbar$-positivity, Narcowich stated necessary and sufficient conditions for a Gaussian to be a Wigner measure:

\begin{lemma} Let ${\bf A}$ be a real, symmetric, positive defined $2 d \times 2d$ matrix. Then the Gaussian
\begin{equation}
f(\xi) = \sqrt{\frac{\det {\bf A}}{\pi^{2 d}}} \exp \left[ - (\xi -\xi_0)^T {\bf A} ( \xi - \xi_0) \right]
\label{Eq2.26}
\end{equation}
is a Wigner measure iff the matrix ${\bf B} = {\bf A}^{-1} + i \hbar {\bf J}$ is a non-negative matrix in $\CO^{2d}$.
\end{lemma}

\noindent
Note that these conditions are equivalent to a Wigner measure satisfying the Robertson-Schr\"odinger form of the uncertainty principle \cite{Gosson1,Gosson2,Robertson,Schrodinger}.

One is also able to tell when a Gaussian represents a pure state \cite{Littlejohn}:

\begin{theorem} {\bf (Littlejohn)} The Gaussian in (\ref{Eq2.26}) is the Wigner measure of a pure state iff there exists a symplectic matrix ${\bf P} \in Sp(2d; \bkR)$ such that ${\bf A}= {\bf P}^T {\bf P}$.
\end{theorem}

\noindent
The NW spectrum of a pure state was completely characterized in \cite{Dias6}:

\begin{theorem} Let $\psi \in L^2 (\bkR^d, dR)$ be a state vector and $f_{\psi} $ the associated Wigner measure. If $\psi$ is a Gaussian, then its NW spectrum reads $\W (f_{\psi}) = \left[- \hbar, \hbar \right]$. If $\psi$ is non-Gaussian, then $\W (f_{\psi}) = \left\{- \hbar, \hbar \right\}$.
\end{theorem}

\end{section}

\begin{section}{Weyl-Wigner formulation of noncommutative quantum mechanics}

In noncommutative quantum mechanics, one replaces the Heisenberg algebra (\ref{Eq2.2}) by an {\it extended Heisenberg algebra}:
\begin{equation}
\left[ \hat z_{\alpha}, \hat z_{\beta} \right] = i \hbar \omega_{\alpha \beta}, \hspace{0.5 cm} \alpha, \beta = 1, \cdots 2d, \hspace{1 cm} {\bf \Omega} = \hbar^{-1}
\left(
\begin{array}{c c}
{\bf \Theta} & \hbar {\bf I}_{d \times d}\\
- \hbar {\bf I}_{d \times d} & {\bf N}
\end{array}
\right),
\label{Eq3.1}
\end{equation}
where $\hat z= (\hat q, \hat p)$ stand for the physical position and momentum variables and ${\bf \Theta}$, ${\bf N}$ are $d \times d$ constant antisymmetric real matrices whose entries $\theta_{ij}$, $\eta_{ij}$, with dimensions $(length)^2$ and $(momentum)^2$, measure the noncommmutativity in the spatial and momentum sectors, respectively. The matrix ${\bf \Omega}$ has entries $\omega_{\alpha \beta}$. We shall tacitly assume that:
\begin{equation}
\theta_{ij} \eta_{kl} < \hbar^2, \hspace{0.5 cm} 1 \le i < j \le d, ~ 1 \le k < l\le d.
\label{Eq3.2}
\end{equation}
This condition, which is compatible with experimental results \cite{Bertolami1,Carroll}, ensures that the skew-symmetric, bilinear form
\begin{equation}
\omega (z,u) = z^T {\bf \Omega} u = z_{\alpha} \omega_{\alpha \beta} u_{\beta}
\label{Eq3.3}
\end{equation}
is non-degenerate (see Lemma 6.1 in the Appendix). By a linear version of Darboux's Theorem, any skew-symmetric bilinear form can be cast in a "normal" form \cite{Cannas,Gosson2} under a linear transformation. This is a sort of symplectic Gram-Schmidt orthogonalization process. We shall denote it by Darboux (D) transformation. In practical terms, since $\Omega$ is even-dimensional and non-degenerate, this means that under the D transformation:
\begin{equation}
\hat z = \hat T (\hat{\xi})= {\bf S} \hat{\xi},
\label{Eq3.4}
\end{equation}
we obtain the Heisenberg algebra (\ref{Eq2.2}) for the variables $\hat{\xi}$. Here ${\bf S}$ is a $2d \times 2d$ constant real matrix. From (\ref{Eq2.2},\ref{Eq3.1},\ref{Eq3.4}), we conclude that:
\begin{equation}
{\bf S} {\bf J} {\bf S}^T = {\bf{\Omega}}.
\label{Eq3.5}
\end{equation}
This implies that $\det {\bf S} = \pm \sqrt{\det {\bf \Omega}}$. But in fact in can be shown that:
\begin{equation}
\det {\bf S} =  \sqrt{\det {\bf{\Omega}}} = \left|Pf ({\bf \Omega}) \right|>0,
\label{Eq3.6}
\end{equation}
where $Pf ({\bf \Omega})$ denotes the Pfaffian of ${\bf \Omega}$. This result is proved in the Appendix.

The D transformation is not unique. Indeed, if ${\bf S}$ is a solution of (\ref{Eq3.5}), then ${\bf S}{\bf L}$ with ${\bf L} \in Sp (2d; \bkR)$ is equally a solution. Nevertheless, in \cite{Bastos1} we proved that all physical predictions (in particular all traces of trace-class operators) are invariant under different choices of D maps. Using this transformation, we constructed a Weyl-Wigner formulation for noncommutative systems by resorting to an {\it extended Weyl-Wigner map}:
\begin{equation}
W_z^{\xi} : \hat{\S}' \longrightarrow \S' (\bkR^{2d}), \hspace{1 cm} \hat A \longrightarrow W_z^{\xi} (\hat A) = T \circ W_{\xi} \circ \hat T^{-1}.
\label{Eq3.7}
\end{equation}
Since the D transformation (\ref{Eq3.4}) is linear, there are no ordering ambiguities and for all practical purposes, $T$ and $\hat T$ are the same transformation: $T(\xi) = {\bf S} \xi$, $\xi \in T^*M$. The extended Weyl-Wigner map is independent of the particular D transformation. This means that:
\begin{equation}
W_z^{\xi} = T \circ W_{\xi} \circ \hat T^{-1} =  T' \circ W_{\xi'} \circ \hat T'^{-1}  = W_z^{\xi'}, \hspace{1 cm} z= T(\xi) = {\bf S} \xi= T' (\xi') = {\bf S'} \xi'
\label{Eq3.8}
\end{equation}
where the matrices ${\bf S}, {\bf S'}$ are both solutions of (\ref{Eq3.5}) and $\hat{\xi}$ and $\hat{\xi}'$ both obey the Heisenberg algebra. Let us denote by $\D_{\Omega} (2d ; \bkR)$ the set of all real $2d \times 2d$ matrices ${\bf S}$ which satisfy (\ref{Eq3.5}). Notice that $\D_{\Omega} (2d ; \bkR)$ is not a subgroup of $Gl (2d ; \bkR)$, as it is not even closed under matrix multiplication.

One of consequences of this definition is that the $\star$-product in noncommutative quantum mechanics becomes \cite{Bastos1}:
\begin{equation}
W_z^{\xi} (\hat A \cdot \hat B ): = W_z^{\xi} (\hat A) \star W_z^{\xi} (\hat B)= A(z) \star B(z) = A (z) \exp \left(\frac{i \hbar}{2} \buildrel{\leftarrow}\over\partial_{z_{\alpha}} \omega_{\alpha \beta}  \buildrel{\rightarrow}\over\partial_{z_{\beta}} \right)B (z),
\label{Eq3.9}
\end{equation}
for $ W_z^{\xi} (\hat A) : = A \in \A (\bkR^{2d})$ and $W_z^{\xi} (\hat B ):= \in \A (\bkR^{2d}) \cup \F$. This $\star$-product admits the kernel representation:
\begin{equation}
A(z) \star B(z) = \frac{1}{(\pi \hbar)^{2d} \det {\bf \Omega}} \int_{\bkR^{2d}} \int_{\bkR^{2d}}  A(z') B (z'') \exp \left[\frac{2i}{\hbar} (z-z')^T {\bf \Omega}^{-1} (z''-z) \right]~ d z' dz'' ,
\label{Eq3.10}
\end{equation}
for $A,B \in \A (\bkR^{2d}) \cup \F$. From (\ref{Eq3.1},\ref{Eq3.9}), we conclude that the $\star$-product is of the form:
\begin{equation}
A(z) \star B(z) = A(z) {\star}_{\hbar} {\star}_{\theta} {\star}_{\eta} B(z)
\label{Eq3.11}
\end{equation}
where
\begin{eqnarray}
A(z) \star_{\hbar} B(z) & = A (z) \exp \left(\frac{i \hbar}{2} \buildrel{\leftarrow}\over\partial_{z_{\alpha}} j_{\alpha \beta}  \buildrel{\rightarrow}\over\partial_{z_{\beta}}\right) B (z),
\label{Eq3.12}\\
A(z) {\star}_{\theta} B(z) & =  A(z) \exp \left(\frac{i}{2} \frac{\buildrel{\leftarrow}\over\partial}{\partial q_i} \theta_{ij} \frac{\buildrel{\rightarrow}\over\partial}{\partial q_j} \right) B(z) \label{Eq3.13}\\
A(z) {\star}_{\eta} B(z) & = A(z) \exp \left(\frac{i}{2} \frac{\buildrel{\leftarrow}\over\partial}{\partial p_i} \eta_{ij} \frac{\buildrel{\rightarrow}\over\partial}{\partial p_j} \right) B(z),
\label{Eq3.14}
\end{eqnarray}
If ${\bf \Theta}$ or ${\bf N}$ are non-degenerate (for instance for $d=2$), we may equally write a kernel representation for $\star_{\theta}, \star_{\eta}$:
\begin{equation}
\begin{array}{l l}
a(q) \star_{\theta} b(q) & = \frac{1}{\pi^d \det {\bf \Theta}} \int_{\bkR^d} \int_{\bkR^d}  a(q') b (q'') \exp \left[2i (q-q')^T {\bf \Theta}^{-1} (q''-q) \right] ~ d q' dq''  ,\\
& \\
c(p) \star_{\eta} d(p) & = \frac{1}{\pi^d \det {\bf N}} \int_{\bkR^d}  \int_{\bkR^d} c(p') d (p'') \exp \left[2i (p-p')^T {\bf N}^{-1} (p''-p) \right] ~ d p' dp'' ,
\end{array}
 \label{Eq3.15}
 \end{equation}
for $a,b \in L^2 (\bkR^d, dq) \cup \A (\bkR^d)$ and $c,d \in L^2 (\bkR^d, dp) \cup \A (\bkR^d)$.

The states of the system are adequately represented by what we called the noncommutative Wigner measures (NCWM) \cite{Bastos1}. The latter can be regarded as the composition of an ordinary Wigner measure with a D transformation (up to a multiplicative normalization constant):
\begin{definition} The noncommutative Wigner measure associated with a state with density matrix $\hat{\rho}$ is a phase-space function of the form
\begin{equation}
f^{NC} (z) : = \frac{1}{(2 \pi \hbar)^d \left| Pf ({\bf \Omega}) \right|}  W_z^{\xi} (\hat{\rho}) = \frac{1}{\left| Pf ({\bf \Omega}) \right|} f^C ({\bf S}^{-1} z) ,
\label{Eq3.16}
\end{equation}
where $f^C$ is the Wigner measure associated with $\hat{\rho}$ and ${\bf S} \in \D_{\Omega} (2d; \bkR)$ is a D transformation. If $f^C$ is the Wigner measure associated with a pure state $\hat{\rho}= | \psi>< \psi|$, then $f^{NC}$ is said to be a pure state noncommutative Wigner measure.
\end{definition}
In particular, one can compute the expectation value of an operator  $\hat A$ in  the state $\hat{\rho}$ according to:
\begin{equation}
 E \left[ \hat A \right] = tr (\hat A \hat{\rho}) = \int_{\bkR^{2d}} f^{NC} (z) A(z) ~ d z ,
\label{Eq3.17}
\end{equation}
where $f^{NC}$ is of the form (\ref{Eq3.16}) with $f^C = (2 \pi \hbar)^{-d} W_{\xi} (\hat{\rho})$ and $A(z) = W_z^{\xi} (\hat A)$.

Here is an alternative characterization of a NCWM which was poved in  \cite{Bastos1}:
\begin{proposition} A function $f(z) \in \F$ is a NCWM iff there exists $b(z) \in \F$ such that:
\begin{eqnarray}
(i) & \int_{\bkR^{2d}} |b(z)|^2 ~dz =1 \label{Eq3.18}\\
(ii) & f(z) = \overline{b(z)} \star b(z) \label{Eq3.19}
\end{eqnarray}
\end{proposition}
With the previous definitions we may prove the following proposition.

\begin{proposition} Let $f^{NC}(z)$ be a NCWM. Then the inequality
\begin{equation}
\int_{\bkR^{2d}} \left( \overline{g(z)} \star g(z) \right) f^{NC} (z) ~dz \ge 0,
\label{Eq3.20}
\end{equation}
holds for any symbol $g(z)$ for which the left-hand side exists.
\end{proposition}
\begin{proof}
In this proof we shall use the cyclic property
\begin{equation}
\int_{\bkR^{2d}} A(z) \star B(z) ~dz=\int_{\bkR^{2d}} A(z)  B(z) ~dz,
\label{Eq3.21}
\end{equation}
which was proved in \cite{Bastos1}. From (\ref{Eq3.19}), we have:
\begin{equation}
\begin{array}{c}
\int_{\bkR^{2d}} \left( \overline{g(z)} \star g(z) \right) f^{NC} (z) ~dz = \int_{\bkR^{2d}} \left( \overline{g(z)} \star g(z) \right) \left( \overline{b(z)} \star b(z) \right) ~dz= \\
\\
= \int_{\bkR^{2d}}  \overline{g(z)} \star g(z) \star \overline{b(z)} \star b(z)  ~dz=\int_{\bkR^{2d}}  b(z) \star \overline{g(z)} \star g(z) \star \overline{b(z)}   ~dz =\\
\\
= \int_{\bkR^{2d}} \left| b(z)\star  \overline{g(z)}  \right|^2  ~dz \ge 0
\end{array}
\label{Eq3.22}
\end{equation}
\end{proof}
\begin{proposition} Any NCWM $f^{NC} (z)$ satisfies the bound:
\begin{equation}
\int_{\bkR^{2d}} \left[f^{NC} (z) \right]^2 ~dz \le \frac{1}{(2 \pi \hbar)^d \left| Pf ({\bf \Omega}) \right|}.
\label{Eq3.23}
\end{equation}
The equality holds iff $f^{NC}$ is a pure state NCWM.
\end{proposition}

\begin{proof} Let us compute the left-hand side of (\ref{Eq3.23}) using (\ref{Eq3.16}):
\begin{equation}
\begin{array}{c}
\int_{\bkR^{2d}} \left[f^{NC} (z) \right]^2 ~dz = \frac{1}{\det {\bf{\Omega}}} \int_{\bkR^{2d}}  \left[f^C ({\bf S}^{-1} z) \right]^2 ~dz  =\\
 \\
 = \frac{\det {\bf S}}{\det {\bf{\Omega}}} \int_{\bkR^{2d}}  \left[f^C (\xi) \right]^2 ~d \xi = \frac{1}{\left| Pf ({\bf \Omega}) \right|} \int_{\bkR^{2d}}  \left[f^C (\xi) \right]^2 ~ d \xi .
\end{array}
\label{Eq3.24}
\end{equation}
In the last step we used (\ref{Eq3.6}). From (\ref{Eq2.15}), the result of the proposition follows immediately.
\end{proof}
Let us now derive the uncertainty principle for NCWMs in the Robertson-Schr\"odinger form \cite{Gosson1,Gosson2,Robertson,Schrodinger}. Let:
\begin{equation}
< \hat z_{\alpha} > : = \int_{\bkR^{2d}} z_{\alpha} f^{NC} (z) ~dz, \hspace{0.5 cm} \alpha = 1, \cdots, 2d
\label{Eq3.25}
\end{equation}
We also define
\begin{equation}
\hat{\tau}_{\alpha}  : = \hat z_{\alpha}  - < \hat z_{\alpha} >, \hspace{0.5 cm} \alpha = 1, \cdots, 2d
\label{Eq3.26}
\end{equation}
Moreover, let ${\bf \Sigma}$ be the covariance matrix with entries:
\begin{equation}
\sigma_{\alpha \beta} = \int_{\bkR^{2d}} \tau_{\alpha} \tau_{\beta} f^{NC} (z) ~dz, \hspace{0.5 cm} \alpha, \beta = 1, \cdots, 2d
\label{Eq3.27}
\end{equation}
We then have:
\begin{proposition}
Let $f^{NC}$ be a NCWM and let ${\bf \Sigma}$ be its covariance matrix. Then it obeys the following uncertainty principle. The matrix
\begin{equation}
{\bf \Sigma} +  \frac{i \hbar}{2} {\bf \Omega}
\label{Eq3.28}
\end{equation}
is non-negative.
\end{proposition}
\begin{proof}
Consider an arbitrary set of $2d$ complex constants $a_{\alpha}$ $(\alpha= 1, \cdots, 2d)$. From (\ref{Eq3.20}), we have:
\begin{equation}
0 \le \int_{\bkR^{2d}} \left( \overline{a_{\alpha}} \tau_{\alpha} \right) \star \left( a_{\beta}\tau_{\beta} \right) f^{NC} (z) ~dz = \overline{a_{\alpha}} \left(\sigma_{\alpha \beta} +  \frac{i \hbar}{2} \Omega_{\alpha \beta} \right) a_{\beta}
\label{Eq3.29}
\end{equation}
which yields the result.
\end{proof}
Such uncertainty relations can be derived in a similar fashion for pairs of noncommuting essentially self-adjoint operators, with some common domain.
\begin{theorem} {\bf (Hudson's Theorem for NCWM)} Let $\psi \in L^2 (\bkR^d, dR)$ be a state vector. The noncommutative Wigner measure of $\psi$ is non-negative iff $\psi$ is a Gaussian state.
\end{theorem}

\begin{proof} The result of the theorem follows from (\ref{Eq3.16}). Indeed $f^{NC} (z)$ is everywhere non-negative iff $f^C (\xi)$ is everywhere non-negative. The rest is an immediate consequence of Theorem 2.1.
\end{proof}

Next, we consider the analogs of Lemma 2.9 and Theorem 2.10.

\begin{lemma} Let ${\bf C}$ be a real, symmetric, positive defined $2d \times 2d$ matrix. Then the Gaussian
\begin{equation}
f(z) = \sqrt{\frac{\det {\bf C}}{\pi^{2d}}} \exp \left( - (z-z_0)^T {\bf C} (z-z_0) \right)
\label{Eq3.30}
\end{equation}
is a NCWM iff the matrix ${\bf D} = {\bf C}^{-1} + i \hbar {\bf \Omega}$ is non-negative in $\CO^{2d}$.
\end{lemma}

\begin{proof} Let us define $g(\xi) = \left|Pf( {\bf \Omega}) \right| f ({\bf S} \xi)$. Then from (\ref{Eq3.30}), we have:
\begin{equation}
g(\xi) = \sqrt{\frac{\det ({\bf C}{\bf \Omega})}{\pi^{2d}}} \exp \left( - ({\bf S} \xi-z_0)^T {\bf C} ({\bf S} \xi-z_0) \right)= \sqrt{\frac{\det {\bf A}}{\pi^{2d}}} \exp \left( - ( \xi- \xi_0)^T {\bf A} ( \xi- \xi_0) \right),
\label{Eq3.31}
\end{equation}
where $\xi_0 = {\bf S}^{-1} z_0$ and ${\bf A} = {\bf S}^T {\bf C} {\bf S}$. Also, from the definition of the matrix ${\bf A}$, it follows that $\det {\bf A} = (\det {\bf S})^2 \det {\bf C}
= \det({\bf C} {\bf \Omega})$. Clearly, ${\bf A}$ is equally a real, symmetric, positive defined matrix. From Lemma 2.9 and the definition (\ref{Eq3.16}) of NCWM, we conclude that $f$ is a NCWM iff the matrix
\begin{equation}
{\bf B}= {\bf A}^{-1} + i \hbar {\bf J} = {\bf S}^{-1} \left( {\bf C}^{-1} + i \hbar {\bf S} {\bf J} {\bf S}^T \right) ({\bf S}^T)^{-1} = {\bf S}^{-1} {\bf D} ({\bf S}^T)^{-1}
\label{Eq3.32}
\end{equation}
is non-negative in $\CO^{2d}$. This is of course equivalent to ${\bf D}$ being non-negative, as the matrices ${\bf S}$ are real.
\end{proof}

Obviously, the Gaussian is simultaneously a Wigner measure and a NCWM iff both ${\bf D}$ and ${\bf C}^{-1} + i \hbar {\bf J}$ are non-negative. Moreover, notice that, since the covariance matrix of the Gaussian is $\frac{1}{2} C^{-1}$, we conclude that the uncertainty principle (Proposition 3.5) is a necessary and sufficient condition for a Gaussian to be a NCWM. For more general states however, it is necessary but not sufficient.

\begin{theorem} {\bf (Littlejohn's Theorem for NCWM)} The Gaussian in (\ref{Eq3.30}) is the NCWM of a pure state iff there exists ${\bf B} \in \D_{\Omega} (2d ; \bkR)$ such that ${\bf C}^{-1} = {\bf B}^T {\bf B}$.
\end{theorem}

\begin{proof} From (\ref{Eq3.16},\ref{Eq3.30}) and Theorem 2.10, we conclude that (\ref{Eq3.30}) is the NCWM of a pure state iff there exists a matrix ${\bf P} \in Sp (2d; \bkR)$ such that ${\bf A} = {\bf P}^T {\bf P}$, i.e. ${\bf C}= \left({\bf B}^{-1} \right)^T {\bf B}^{-1}$, with ${\bf B}^{-1}= {\bf P} {\bf S}^{-1}$. From (\ref{Eq3.5}), we conclude that ${\bf B}^{-1} {\bf \Omega } \left({\bf B}^{-1} \right)^T = {\bf J}$. In other words: ${\bf B} \in \D_{\Omega} (2d; \bkR)$.
\end{proof}

Before we conclude this section let us briefly discuss the set of diffeomorphisms that transform a NCWM into another NCWM.

\begin{definition} We denote by $Sp_{\Omega} (2d; \bkR)$ the set of noncommutative symplectic transformations. These are the $2d \times 2d$ real, constant matrices ${\bf M}$ that satisfy:
\begin{equation}
{\bf M} {\bf \Omega} {\bf M}^T= {\bf \Omega}.
\label{Eq3.33}
\end{equation}

These correspond to the linear transformations which leave the skew-symmetric bilinear form (\ref{Eq3.3}) invariant. It is not very difficult to check that $Sp_{\Omega} (2d; \bkR)$ is a group with respect to matrix multiplication. We also define the group of noncommutative symplectomorphisms to be the set of automorphisms $\phi$ of $\bkR^{2d}$ such that $d \phi (z) \in Sp_{\Omega} (2d; \bkR)$ for all $z \in \bkR^{2d}$.
\end{definition}

\begin{lemma} The groups $Sp_{\Omega} (2d; \bkR)$ and $Sp (2d; \bkR)$ are isomorphic.
\end{lemma}

\begin{proof} Fix some element ${\bf S} \in \D_{\Omega} (2d; \bkR)$. Then the map
\begin{equation}
\begin{array}{l c l}
\phi_{{\bf S}}: & Sp (2d; \bkR) & \longrightarrow Sp_{\Omega} (2d; \bkR)\\
& & \\
& {\bf P} & \longmapsto \phi_{{\bf S}} ({\bf P}) = {\bf S} {\bf P} {\bf S}^{-1}
\label{Eq3.34}
\end{array}
\end{equation}
is a Lie group isomorphism.
\end{proof}

\begin{lemma} Noncommutative symplectic transformations map NCWM's to NCWM's.
\end{lemma}

\begin{proof} Let $M \in Sp_{\Omega} (2d; \bkR)$ and let $f^{NC}$ be some NCWM. Then there exist a Wigner measure $f^C$ and ${\bf S} \in \D_{\Omega} (2d; \bkR)$ such that (\ref{Eq3.16}) holds. Under the noncommutative symplectic transformation $z \longrightarrow {\bf M} z$, we obtain
$$f^{NC}(z) \longrightarrow f^{NC} ({\bf M}z) = \frac{1}{\left|Pf ( {\bf \Omega}) \right|} f^C \left({\bf S}^{-1} {\bf M} z \right)= \frac{1}{\left|Pf ( {\bf \Omega}) \right|} f^C \left(\left( {\bf M}^{-1}{\bf S} \right)^{-1} z \right).
$$
On the other hand $\left({\bf M}^{-1} {\bf S} \right) {\bf J} \left({\bf M}^{-1} {\bf S} \right)^T = {\bf \Omega}$, which means that ${\bf M}^{-1} {\bf S}  \in \D_{\Omega} (2d; \bkR)$. It follows that $f^{NC} ({\bf M}z)$ is again a NCWM.
\end{proof}

\begin{lemma} The uncertainty principle is invariant under noncommutative symplectic transformations.
\end{lemma}
\begin{proof}
Let ${\bf M} \in Sp_{\Omega} (2d; \bkR)$ and $\hat z' = {\bf M} \hat z$. We then have $\hat{\tau}' = {\bf M} \hat{\tau}$ and ${\bf \Sigma}' = {\bf M} {\bf \Sigma} {\bf M}^T$. For arbitrary $a \in \CO^{2d}$:
\begin{equation}
a^{\dagger} \left({\bf \Sigma}' + \frac{i \hbar}{2} {\bf \Omega} \right) a =a^{\dagger} \left({\bf M} {\bf \Sigma} {\bf M}^T + \frac{i \hbar}{2} {\bf \Omega} \right) a =\left(M^T a\right) ^{\dagger} \left({\bf \Sigma}  + \frac{i \hbar}{2} {\bf \Omega} \right) \left(M^T a \right) \ge 0
\label{Eq3.35}
\end{equation}
\end{proof}

\end{section}

\begin{section}{Properties of noncommutative Wigner measures in two dimensions}

In this section we shall focus on the $d=2$ case. The extended Heisenberg algebra in two dimensions reads:
\begin{equation}
\left[ \hat z_{\alpha}, \hat z_{\beta} \right] = i \hbar \omega_{\alpha \beta}, \hspace{0.5 cm} \alpha, \beta = 1, \cdots 4, \hspace{1 cm}
{\bf{\Omega}} = \left(
\begin{array}{c c}
\frac{\theta}{\hbar} {\bf E} & {\bf I}_{2 \times 2}\\
- {\bf I}_{2 \times 2} & \frac{\eta}{\hbar} {\bf E}
\end{array}
\right).
\label{Eq4.1}
\end{equation}
The $2\times 2$ matrix ${\bf E}$ has entries $\epsilon_{11}= \epsilon_{22}=0$, $\epsilon_{12}= - \epsilon_{21} =1$. The real, constant, positive parameters $\theta, \eta$ measure the noncommutativity in phase-space. In this case, we have from (\ref{Eq3.2},\ref{Eq3.6}):
\begin{equation}
\det {\bf S} = \left|Pf ( {\bf \Omega}) \right| = 1  - \zeta, \hspace{1 cm} \zeta = \frac{\theta \eta}{\hbar^2} <1, \label{Eq4.2}
\end{equation}
where ${\bf S} \in \D_{\Omega} (4; \bkR)$. In particular, the bound on the purity (\ref{Eq3.23}) reads:
\begin{equation}
\int_{\bkR^4}\left[f^{NC} (z) \right]^2 ~ dz \le \frac{1}{(2 \pi \sqrt{1 - \zeta} )^2}.
\label{Eq4.3}
\end{equation}
In this paper we shall frequently use the following D map:
\begin{equation}
{\bf S} = \left(
\begin{array}{c c}
\lambda {\bf I}_{2 \times 2} & - \frac{\theta}{2 \lambda \hbar} {\bf E}\\
\frac{\eta}{2 \mu \hbar} {\bf E} & \mu {\bf I}_{2 \times 2}
\end{array}
\right),
\label{Eq4.4}
\end{equation}
where $\lambda, \mu$ are adimensional real constants such that:
\begin{equation}
\mu \lambda = \frac{1 + \sqrt{1 - \zeta}}{2}
\label{Eq4.5}
\end{equation}
From (\ref{Eq4.2},\ref{Eq4.4}), we conclude that the D map admits the inverse:
\begin{equation}
{\bf S}^{-1} = \frac{1}{\sqrt{1 - \zeta}} \left(
\begin{array}{c c}
\mu {\bf I}_{2 \times 2} &  \frac{\theta}{2 \lambda \hbar} {\bf E}\\
- \frac{\eta}{2 \mu \hbar} {\bf E} & \lambda {\bf I}_{2 \times 2}
\end{array}
\right),
\label{Eq4.6}
\end{equation}

\begin{subsection}{Marginal distributions}

Our aim is now to derive some properties of NCWMs for $d=2$. Let us start by defining the position and momentum marginal distributions:

\begin{definition} The position and momentum marginal distributions of a NCWM $f^{NC}$ are defined by:
\begin{equation}
\P_q (q) = \int_{\bkR^2} f^{NC} (q,p) ~dp \hspace{1 cm} \P_p (p) = \int_{\bkR^2}  f^{NC} (q,p) ~dq .
\label{Eq4.1.1}
\end{equation}
\end{definition}

It is important to remark that, albeit real and normalized, these distributions cannot be regarded as true probability densities as was the case with the commutative counterparts. This is not surprising given that the position variables do not commute amongst themselves and neither do the momenta. Therefore we cannot expect to obtain a joint probability density for, say, $q_1$ and $q_2$. This manifests itself in the fact that $\P_q$ and $\P_p$ may (and usually do) take on negative values. This is an immediate consequence of the following theorem:

\begin{theorem} The position and momentum marginal distributions of the NCWM associated with the pure state $\hat{\rho} = | \psi>< \psi|$ can be written as:
\begin{eqnarray}
\P_q (q) = \phi (q) \star_{\theta} \overline{\phi (q)} \label{Eq4.1.2}\\
\P_p (p) = \chi (p) \star_{\eta} \overline{\chi (p)} \label{Eq4.1.3}
\end{eqnarray}
where
\begin{equation}
\phi (q) = \frac{1}{\lambda} \psi \left( \frac{q}{\lambda} \right) \hspace{1 cm} \chi (p) = \frac{1}{\mu} \hat{\psi} \left( \frac{p}{\mu} \right).
\label{Eq4.1.4}
\end{equation}
Here $\lambda, \mu$ are the adimensional parameters that appear in the D transformation (\ref{Eq4.4}) and $\hat{\psi} (\Pi)$ is the Fourier transform of $\psi (R) = < R| \psi>$.
\end{theorem}

\begin{proof} The NCWM of a pure state is given by (\ref{Eq3.16}) with (cf.(\ref{Eq2.13})):
\begin{equation}
f^C (R, \Pi) = \frac{1}{( \pi \hbar)^2} \int_{\bkR^2}  e^{ -2 i y \cdot \Pi / \hbar} \overline{\psi (R-y)} \psi (R+y)~dy .
\label{Eq4.1.5}
\end{equation}
We then have:
\begin{equation}
\P_q (q) = \frac{1}{( \pi \hbar \sqrt{1- \zeta})^2} \int_{\bkR^2}  \int_{\bkR^2} e^{ -2 i y \cdot \Pi (q,p) / \hbar} \overline{\psi \left[R(q,p) -y \right]} \psi \left[R(q,p) +y \right] ~dy dp .
\label{Eq4.1.6}
\end{equation}
We next perform the change of variables:
\begin{equation}
u = \lambda \left[R(q,p) +y \right] , \hspace{1 cm} v = \lambda \left[R(q,p) -y \right]
\label{Eq4.1.7}
\end{equation}
and obtain:
\begin{equation}
\begin{array}{c}
\P_q (q) = \frac{1}{( \pi \theta)^2} \int_{\bkR^2}\int_{\bkR^2}   \phi (u) \overline{\phi (v)} \times \\
\\
\times \exp\left\{ - \frac{i}{\hbar \sqrt{1 - \zeta}} (u-v) \cdot \left[\left(\frac{2 \hbar \mu \lambda}{\theta}  - \frac{\eta}{2 \mu \lambda \hbar} \right) {\bf E} q  - \frac{\hbar}{\theta}  \sqrt{1- \zeta} {\bf E} (u+v) \right] \right\} ~ du dv  .
\end{array}
\label{Eq4.1.8}
\end{equation}
Here $\phi (q)$ is given by (\ref{Eq4.1.4}). Notice also that from (\ref{Eq4.5}), we have:
\begin{equation}
\frac{2 \hbar \mu \lambda}{\theta}  - \frac{\eta}{2 \mu \lambda \hbar} = \frac{2 \hbar}{\theta} \sqrt{1 - \zeta}.
\label{Eq4.1.9}
\end{equation}
Consequently:
\begin{equation}
\P_q (q) = \frac{1}{( \pi \theta)^2} \int_{\bkR^2} \int_{\bkR^2}  \phi (u) \overline{\phi (v)} \exp\left[ - \frac{2i}{\theta} (u-q)^T {\bf E} (q-v) \right] ~du dv ,
\label{Eq4.1.10}
\end{equation}
which is the kernel representation of the $\star_{\theta}$-product (\ref{Eq3.15}). Likewise, let us consider the Wigner measure (\ref{Eq4.1.5}) in the momentum representation:
\begin{equation}
f^C (R, \Pi) = \frac{1}{( \pi \hbar)^2} \int_{\bkR^2} e^{ -2 i k \cdot R / \hbar} \overline{\hat{\psi} (k-\Pi)} \hat{\psi} (k+\Pi) ~ dk .
\label{Eq4.1.11}
\end{equation}
Following exactly the same procedure, we arrive at the kernel representation of $\chi (p) \star_{\eta} \overline{\chi (p)}$, where $\chi (p)$ is given by (\ref{Eq4.1.4}).
\end{proof}

The nonlocal nature of the $\star$-products in (\ref{Eq4.1.2},\ref{Eq4.1.3}) entails that the marginal distributions may take on negative values. Moreover, Theorem 4.2 has another important consequence. If we look at eq.(\ref{Eq4.1.2}), we notice that $\P_q (q)$ is a positive, normalized element of the $\star$-algebra with $\star_{\theta}$-product. Furthermore, the matrix ${\bf E}$ coincides with the sympletic matrix ${\bf J}$ of a two-dimensional phase space. We thus conclude that, if we make the correspondence ${\bf E} \leftrightarrow  {\bf J}$, $q_1 \leftrightarrow R$, $q_2 \leftrightarrow \Pi$, $\theta \leftrightarrow \hbar$ for $d=1$, then the expression (\ref{Eq4.1.2}) states that $\P_q (q)$ can be written as (cf.(\ref{Eq2.13},\ref{Eq2.16},\ref{Eq2.17})):
\begin{equation}
\P_q (q) = \sum_j s_j f_j^{\theta} (q_1, q_2),
\label{Eq4.1.12}
\end{equation}
with
\begin{equation}
0 \le s_j \le 1,~~ \forall j, \hspace{1 cm} \sum_j s_j =1,
\label{Eq4.1.13}
\end{equation}
and
\begin{equation}
f_j^{\theta} (q_1, q_2) = \frac{1}{\pi \theta} \int_{\bkR} e^{-2 i y q_2 / \theta} \overline{\phi_j (q_1 -y)} \phi_j (q_1 + y) ~ d y,
\label{Eq4.1.14}
\end{equation}
for some set of wave functions $\phi_j \in L^2 ( \bkR, dq_1)$ which can be chosen to be orthonormal:
\begin{equation}
\int_{\bkR} \overline{\phi_j (q_1)} \phi_k (q_1)~d q_1  = \delta_{jk}.
\label{Eq4.1.15}
\end{equation}
We shall call functions of the form (\ref{Eq4.1.12}-\ref{Eq4.1.14}) {\bf $\theta$-Wigner measures}.

If instead of a pure state we consider a mixed state:
\begin{equation}
f^{NC} (z) = \sum_i p_i f_i^{NC} (z),
\label{Eq4.1.16}
\end{equation}
then eq.(\ref{Eq4.1.2}) would be replaced by:
\begin{equation}
\P_q (q) = \sum_i p_i \phi_i (q) \star_{\theta} \overline{\phi_i (q)}.
\label{Eq4.1.17}
\end{equation}
This is then a convex combination of $\theta$-Wigner measures, which means that it must again be of the form (\ref{Eq4.1.12}). Analogous conclusions can be drawn for the momentum marginal distribution.
\begin{equation}
\P_p (p) = \sum_k r_k f_k^{\eta} (p_1, p_2),
\label{Eq4.1.18}
\end{equation}
with
\begin{equation}
0 \le r_k \le 1, ~~ \forall k, \hspace{1 cm} \sum_k r_k =1,
\label{Eq4.1.19}
\end{equation}
and
\begin{equation}
f_k^{\eta} (p_1, p_2) = \frac{1}{\pi \eta} \int_{\bkR}e^{-2 i y p_2 / \eta} \overline{\chi_k (y-p_1 )} \chi_k (y+p_1 ) ~ d y .
\label{Eq4.1.20}
\end{equation}
The functions $\chi_k \in L^2 ( \bkR, dp_1)$ can also be chosen to be orthonormal:
\begin{equation}
\int_{\bkR}  \overline{\chi_k (p_1)} \chi_l (p_1)~ d p_1  = \delta_{kl}.
\label{Eq4.1.21}
\end{equation}
The functions of the form (\ref{Eq4.1.18}-\ref{Eq4.1.20}) are called {\bf $\eta$-Wigner measures}.

From (\ref{Eq2.15},\ref{Eq4.1.12},\ref{Eq4.1.18}) we conclude that:

\begin{theorem} The marginal distributions of a NCWM satisfy the bounds:
\begin{eqnarray}
\int_{\bkR^2}  \left[\P_q (q) \right]^2 ~d q \le \frac{1}{2 \pi \theta} \label{Eq4.1.22} \\
\int_{\bkR^2} \left[\P_p (p) \right]^2 ~ d p \le \frac{1}{2 \pi \eta} \label{Eq4.1.23}
\end{eqnarray}
\end{theorem}

These bounds have no counterpart in ordinary quantum mechanics. We shall denote the integrals in (\ref{Eq4.1.22},\ref{Eq4.1.23}) by {\bf $\theta$-purity} and {\bf $\eta$-purity}, respectively. The integral (\ref{Eq4.3}) will be simply called purity. An interesting question regards the saturation of the inequalities (\ref{Eq4.1.22},\ref{Eq4.1.23}). The following lemma shows that it is possible to maximize them.

\begin{lemma} The states of the form:
\begin{equation}
f^{NC} (q,p) = \frac{1}{1 - \zeta} \left( \frac{\theta}{\hbar} \right)^2 f_{\phi_1}^{\theta} \left( \frac{q + (\theta / \hbar) {\bf E} p}{\sqrt{1 - \zeta}} \right) f_{\phi_2}^{\theta} (q)
\label{Eq4.1.24}
\end{equation}
where $f_{\phi_i}^{\theta} (q)$ is the $\theta$-Wigner measure associated with the function $\phi_i \in L^2 (\bkR, dq_1)$ $(i=1,2)$, maximize the $\theta$-purity (\ref{Eq4.1.22}). Likewise, the states of the form:
\begin{equation}
f^{NC} (q,p) = \frac{1}{1 - \zeta} \left( \frac{\eta}{\hbar} \right)^2 f_{\chi_1}^{\eta} \left( \frac{p + (\eta / \hbar) {\bf E} q}{\sqrt{1 - \zeta}} \right) f_{\chi_2}^{\eta} (p)
\label{Eq4.1.25}
\end{equation}
where $f_{\chi_i}^{\eta} (q)$ is the $\eta$-Wigner measure associated with the function $\chi_i \in L^2 (\bkR, dp_1)$ $(i=1,2)$, maximize the $\eta$-purity (\ref{Eq4.1.23}).
\end{lemma}

\begin{proof} First of all, let us check that (\ref{Eq4.1.24}) is indeed a NCWM. Consider the normalized wavefunction:
\begin{equation}
\psi (q) = \lambda \sqrt{\frac{2}{\pi \theta}} \int_{\bkR^2} e^{-2i \lambda x q_2 / \theta} \overline{\phi_1 (\lambda q_1 -x)} \phi_2 (\lambda q_2 + x )~d x ,
\label{Eq4.1.26}
\end{equation}
where $\lambda$ is as in eq.(\ref{Eq4.4}). The corresponding NCWM reads:
\begin{equation}
\begin{array}{c}
f^{NC} (q,p) = \left(\frac{\lambda}{\pi \hbar \sqrt{1- \zeta}} \right)^2 \frac{2}{\pi \theta} \int_{\bkR} \int_{\bkR} \int_{\bkR} \int_{\bkR}
\exp \left(- \frac{2i}{\hbar} y_1 \Pi_1 - \frac{2i}{\hbar} y_2 \Pi_2 +  \right. \\
\\
\left. +\frac{2i \lambda}{\theta} x_1 (R_2 - y_2)- \frac{2i \lambda}{\theta} x_2 (R_2 + y_2) \right)
\times \overline{\phi_1 \left( \lambda (R_1 + y_1 ) - x_2 \right)} \phi_1 \left( \lambda (R_1 - y_1 ) - x_1 \right) \\
\\ \overline{\phi_2 \left( \lambda (R_1 - y_1 ) + x_1 \right)} \phi_2 \left( \lambda (R_1 + y_1 ) + x_2 \right) ~ d y_1 dy_2  d x_1 d x_2.
\label{Eq4.1.27}
\end{array}
\end{equation}
Upon integration over $y_2, x_2$, we obtain:
\begin{equation}
\begin{array}{c}
f^{NC} (q,p) = \frac{2 \lambda}{\left( \pi \hbar \sqrt{1- \zeta} \right)^2 } \int_{\bkR} \int_{\bkR}  \exp \left[- \frac{2i}{\hbar} y_1 \Pi_1 + \frac{2i \lambda}{\theta} R_2 \left(2 x_1 + \frac{\theta}{\hbar \lambda} \Pi_2 \right) \right]\times \\
\\
\times \overline{\phi_1 \left( \lambda (R_1 + y_1 ) + x_1 + \frac{\theta}{\hbar \lambda} \Pi_2 \right)} \phi_1 \left( \lambda (R_1 - y_1 ) - x_1 \right) \times \\
\\
\times \overline{\phi_2 \left( \lambda (R_1 - y_1 ) + x_1 \right)} \phi_2 \left( \lambda (R_1 + y_1 ) - x_1 - \frac{\theta}{\hbar \lambda} \Pi_2 \right) ~ d y_1  d x_1  .
\label{Eq4.1.28}
\end{array}
\end{equation}
We then perform the substitution
\begin{equation}
y_1 = \frac{v-u}{2 \lambda}, \hspace{1 cm} x_1 = - \left( \frac{u+v}{2} \right) - \frac{\theta}{2 \hbar \lambda} \Pi_2,
\label{Eq4.1.29}
\end{equation}
with Jacobian $(2 \lambda)^{-1}$. The result is:
\begin{equation}
f^{NC} (q,p) = \left(\frac{\theta}{\hbar} \right)^2 \frac{1}{1- \zeta} f_{\phi_1}^{\theta} \left( \lambda R + \frac{\theta}{2 \hbar \lambda} {\bf E} \Pi \right) f_{\phi_2}^{\theta} \left( \lambda R - \frac{\theta}{2 \hbar \lambda} {\bf E} \Pi \right).
\label{Eq4.1.30}
\end{equation}
Taking into account the D transformation (\ref{Eq4.4}), we obtain (\ref{Eq4.1.24}). By construction this is indeed a NCWM. It remains to prove that (\ref{Eq4.1.24}) maximizes the $\theta$-purity (\ref{Eq4.1.22}).
\begin{equation}
\P_q (q) = \left(\frac{\theta}{\hbar} \right)^2 \frac{1}{1- \zeta} f_{\phi_2}^{\theta} (q)  \int_{\bkR^2}f_{\phi_1}^{\theta} \left(\frac{q + (\theta / \hbar) {\bf E} p}{\sqrt{1 - \zeta}} \right)~ d p  = f_{\phi_2}^{\theta} (q).
\label{Eq4.1.31}
\end{equation}
We then have:
\begin{equation}
\int_{\bkR^2}  \left[ \P_q (q) \right]^2 ~dq = \int_{\bkR^2} \left[f_{\phi_2}^{\theta} (q) \right]^2 ~dq = \frac{1}{2 \pi \theta}.
\label{Eq4.1.32}
\end{equation}
In the last step, we used the fact that $f_{\phi_2}^{\theta}$ is a $\theta$-Wigner measure (\ref{Eq4.1.14}). The proof that the states (\ref{Eq4.1.25}) saturate the $\eta$-purity is analogous.
\end{proof}

\begin{lemma} States of the form (\ref{Eq4.1.24},\ref{Eq4.1.25}) cannot be Wigner measures.
\end{lemma}

\begin{proof} Let us compute the purity of the states of the form (\ref{Eq4.1.24}):
\begin{equation}
\begin{array}{c}
\int_{\bkR^2} \int_{\bkR^2}  \left[f^{NC} (q,p) \right]^2 ~dq  dp= \left(\frac{\theta}{\hbar \sqrt{1 - \zeta}} \right)^4 \int_{\bkR^2} \int_{\bkR^2} \left[f_{\phi_1}^{\theta} \left( \frac{q + (\theta / \hbar) {\bf E} p}{\sqrt{1- \zeta}} \right) \right]^2 \left[f_{\phi_2}^{\theta} (q ) \right]^2~dq  dp =\\
\\
= \left(\frac{\theta}{\hbar \sqrt{1 - \zeta}} \right)^2 \int_{\bkR^2} \int_{\bkR^2}  \left[f_{\phi_1}^{\theta} (q' ) \right]^2 \left[f_{\phi_2}^{\theta} (q ) \right]^2 ~dq d q'= \frac{1}{(2 \pi \hbar \sqrt{1 - \zeta})^2} > \frac{1}{(2 \pi \hbar )^2}
\label{Eq4.1.33}
\end{array}
\end{equation}
We obtain the same result if we use instead the states (\ref{Eq4.1.25}). From (\ref{Eq2.15}) this proves the lemma.
\end{proof}

Notice that we may apply Hudson's theorem to the marginal distributions whenever we have states of the form (\ref{Eq4.1.24}) or (\ref{Eq4.1.25}): the marginal distributions $\P_q (q)$ or $\P_p (p)$ are positive iff $\phi_2$ or $\chi_2$ are Gaussian, respectively.

Before we conclude this section, a brief remark is in order. States of the form (\ref{Eq4.1.24},\ref{Eq4.1.25}) look like the tensor product of one-dimensional states. It then seems natural to look for entangled states of the form
\begin{equation}
\begin{array}{c}
f^{NC} (q,p) = \frac{1}{1 - \zeta} \left( \frac{\theta}{\hbar} \right)^2 \left[r f_{\phi_1}^{\theta} \left( \frac{q + (\theta / \hbar) {\bf E} p}{\sqrt{1 - \zeta}} \right) f_{\phi_2}^{\theta} (q) + \right.\\
\\
\left. + (1-r)  f_{\psi_1}^{\theta} \left( \frac{q + (\theta / \hbar) {\bf E} p}{\sqrt{1 - \zeta}} \right) f_{\psi_2}^{\theta} (q) \right],
\end{array}
\label{Eq4.1.40}
\end{equation}
with $0 <r<1$. As we mentioned in the introduction, this may be the starting point for a theory of noncommutative quantum information and quantum computation for continuous variables \cite{Giedke}. Hopefully, this may lead to qualitatively new predictions that could signal the existence of noncommutativity in the physical world. This will be the subject of a future work.

\end{subsection}

\begin{subsection}{Noncommutative Narcowich-Wigner spectrum}

An interesting question is that of constructing the analog of the KLM conditions for noncommutative Wigner measures. Let us start by defining the noncommutative version of the symplectic Fourier transform:

\begin{definition} Let $f \in \F$. Its noncommutative symplectic Fourier transform is defined by:
\begin{equation}
\tilde f^{{\bf \Omega}} (a) = \int_{\bkR^4}  f(z) \exp  \left(i a^T {\bf \Omega}^{-1} z \right) ~dz.
\label{Eq4.2.1}
\end{equation}
The inverse formula is:
\begin{equation}
f(z) = \frac{1}{(2 \pi \sqrt{1- \zeta} )^4}  \int_{\bkR^4}  \tilde f^{{\bf \Omega}} (a) \exp  \left(- i a^T {\bf \Omega}^{-1} z \right) ~da.
\label{Eq4.2.2}
\end{equation}
\end{definition}

As usual we define the convolution of $f \in L_p$ and $g \in L_q$ with $p^{-1}+q^{-1}=1$:
\begin{equation}
(f \natural g) (z) = \int_{\bkR^4} f(z-z') g (z') d z'= \int_{\bkR^4}  f(z') g (z-z') ~dz'.
\label{Eq4.2.3}
\end{equation}
From (\ref{Eq4.2.1}-\ref{Eq4.2.3}) it is easy to prove that the noncommutative symplectic Fourier transform of the convolution amounts to pointwise multiplication:
\begin{equation}
\widetilde{(f \natural g)}^{{\bf \Omega}} = \tilde f^{{\bf \Omega}} (a) \cdot \tilde g^{{\bf \Omega}} (a).
\label{Eq4.2.4}
\end{equation}
We then propose the following noncommutative generalization of Definition 2.3:

\begin{definition} We say that $\tilde f^{{\bf \Omega}} (a)$ is of the $(\alpha, \beta, \gamma )$-positive type, if the $m \times m$ matrix with entries
\begin{equation}
N_{jk} \equiv \tilde f^{{\bf \Omega}} (a_j-a_k) \exp \left( \frac{i}{2} a_k^T {\bf \Lambda} (\alpha, \beta, \gamma ) a_j \right)
\label{Eq4.2.5}
\end{equation}
is hermitian and non-negative for any positive integer $m$ and any set of $m$ points $a_1, \cdots, a_m$ in the dual of the phase space. The matrix ${\bf \Lambda}$ is defined by:
\begin{equation}
{\bf \Lambda} (\alpha, \beta, \gamma ) = \left(
\begin{array}{c c}
\gamma {\bf E} & - \alpha {\bf I}_{2 \times 2}\\
\alpha {\bf I}_{2 \times 2} & \beta {\bf E}
\end{array}
\right).
\label{Eq4.2.6}
\end{equation}
\end{definition}

Then the following theorem holds:

\begin{theorem} The function $f \in \F$ is a NCWM iff its noncommutative symplectic Fourier transform $\tilde f^{{\bf \Omega}} (a)$ satisfies the set of noncommutative KLM conditions:
\begin{eqnarray}
(i) & \tilde f^{{\bf \Omega}} (0) & =1 \label{Eq4.2.7}\\
(ii) & \tilde f^{{\bf \Omega}} (a) & {\mbox{is continuous and of the $(\tilde{\hbar}, \tilde{\theta}, \tilde{\eta})$-positive type}} \label{Eq4.2.8}
\end{eqnarray}
where
\begin{equation}
\tilde{\hbar}= (1 - \zeta)^{-1} \hbar, \hspace{0.3 cm} \tilde{\theta}= (1 - \zeta)^{-1} \theta, \hspace{0.3 cm} \tilde{\eta}= (1 - \zeta)^{-1} \eta. \label{Eq4.2.9}
\end{equation}
\end{theorem}

\begin{proof} The function $f \in \F$ is a NCWM iff there exists a Wigner measure $g$ and a matrix ${\bf S} \in \D_{\Omega} (2d; \bkR)$ such that
\begin{equation}
f(z) = \frac{1}{\left| Pf({\bf \Omega}) \right|} g ( {\bf S}^{-1} z). \label{Eq4.2.10}
\end{equation}
The function $g(\xi)$ thus satisfies the KLM conditions (\ref{Eq2.22},\ref{Eq2.23}). Let us compute the noncommutative symplectic Fourier transform of (\ref{Eq4.2.10}). From (\ref{Eq2.19},\ref{Eq3.5},\ref{Eq3.6}) and the fact that ${\bf J}^{-1} = {\bf J}^T = - {\bf J}$, we get:
\begin{equation}
\begin{array}{c}
\tilde f^{{\bf \Omega}} (a) = \frac{1}{\left| Pf({\bf \Omega}) \right|} \int_{\bkR^4}  g ( {\bf S}^{-1} z) \exp \left( i a^T {\bf \Omega}^{-1} z \right) ~dz= \\
\\
= \frac{\det {\bf S}}{\left| Pf({\bf \Omega}) \right|} \int_{\bkR^4}  g ( \xi) \exp \left( i a^T {\bf \Omega}^{-1} {\bf S} \xi \right)~d \xi = \\
\\
=
 \int_{\bkR^4} g ( \xi) \exp \left( - i({\bf S}^{-1} a)^T {\bf J} \xi \right) ~d \xi= \tilde g^{{\bf J}} ({\bf S}^{-1} a ). \label{Eq4.2.11}
\end{array}
\end{equation}
From this equation it follows that conditions (\ref{Eq2.22}) and (\ref{Eq4.2.7}) are equivalent. For any positive integer $m$ and any set of points $a_1, \cdots, a_m$ in the dual of the phase-space, let us consider the matrices:
\begin{equation}
M_{jk} = \tilde g^{{\bf J}} (a_j -a_k) \exp \left( - \frac{i \hbar}{2} a_k^T {\bf J} a_j \right).
\label{Eq4.2.12}
\end{equation}
If we define $b_i = {\bf S} a_i$ $(i=1, \cdots, m)$, we get from (\ref{Eq4.2.11}):
\begin{equation}
M_{jk} = \tilde f^{{\bf \Omega}} (b_j -b_k) \exp \left( \frac{i \hbar}{2} b_k^T {\bf \Omega}^{-1} b_j \right),
\label{Eq4.2.13}
\end{equation}
where we used (cf.(\ref{Eq3.5})):
\begin{equation}
- ({\bf S}^T)^{-1} {\bf J} {\bf S}^{-1} = ({\bf S}{\bf J}{\bf S}^T )^{-1} = {\bf \Omega}^{-1}.
\label{Eq4.2.14}
\end{equation}
Now notice that:
\begin{equation}
\frac{i \hbar}{2} {\bf \Omega}^{-1} = \frac{i}{2} \left(
\begin{array}{c c}
\tilde{\eta} {\bf E} & - \tilde{\hbar} {\bf I}_{2 \times 2}\\
\tilde{\hbar} {\bf I}_{2 \times 2} & \tilde{\theta} {\bf E}
\end{array}
\right) = \frac{i}{2} {\bf \Lambda} (\tilde{\hbar}, \tilde{\theta} ,\tilde{\eta}).
\label{Eq4.2.15}
\end{equation}
That is:
\begin{equation}
M_{jk} = \tilde f^{{\bf \Omega}} (b_j -b_k) \exp \left( \frac{i}{2} b_k^T {\bf \Lambda} (\tilde{\hbar}, \tilde{\theta} ,\tilde{\eta}) b_j \right),
\label{Eq4.2.16}
\end{equation}
The function $\tilde g^{{\bf J}} (a)$ is then of the $\hbar$-positive type, iff $\tilde f^{{\bf \Omega}} (a)$ is of the $(\tilde{\hbar}, \tilde{\theta} ,\tilde{\eta})$-positive type.
\end{proof}

\begin{definition} The noncommutative Narcowich-Wigner (NCNW) spectrum of $f \in \F$ is the set:
\begin{equation}
\W^{NC} (f) = \left\{(\alpha, \beta, \gamma) \in \bkR^3 \left| \tilde f^{{\bf \Omega}} (a) {\mbox{ is of $(\alpha, \beta, \gamma)$-positive type}} \right. \right\}.
\label{Eq4.2.17}
\end{equation}
Obviously, if $f$ is a NCWM, then $(\tilde{\hbar},\tilde{\theta},\tilde{\eta}) \in \W^{NC} (f)$.
\end{definition}

From this definition, we may prove the analog of Theorem 2.8:

\begin{theorem} The NCNW spectrum of the convolution $f \natural g$ of $f,g \in \F$ contains all elements of the form
\begin{equation}
(\alpha_1 + \alpha_2, \beta_1 + \beta_2, \gamma_1 + \gamma_2 ),
\label{Eq4.2.18}
\end{equation}
with $(\alpha_1 , \beta_1 , \gamma_1) \in \W^{NC} (f)$ and $(\alpha_2 , \beta_2 , \gamma_2) \in \W^{NC} (g)$.
\end{theorem}

\begin{proof} The theorem follows immediately from the facts that (i) the Schur (or Hadamard) product of two hermitian and non-negative matrices is again a hermitian, non-negative matrix, and (ii) the matrix ${\bf \Lambda} (\alpha, \beta, \gamma)$ is additive with respect to its arguments:
\begin{equation}
{\bf \Lambda} (\alpha_1+ \alpha_2, \beta_1 + \beta_2, \gamma_1+\gamma_2) = {\bf \Lambda} (\alpha_1, \beta_1 , \gamma_1) + {\bf \Lambda} ( \alpha_2, \beta_2, \gamma_2).
\label{Eq4.2.19}
\end{equation}
\end{proof}

The interpretation of $(\alpha, \beta, \gamma)$-positive functions is slightly less straightforward than that of $\alpha$-positive functions in the commutative case. Indeed, we know from Theorem 2.6, that that if $g(\xi)$ is properly normalized and if its symplectic Fourier transform is continuous and of the $\alpha$-positive type (with $\alpha \ne 0$), then we can just think of it as a Wigner function, where $\hbar$ has been replaced by $\alpha$. In other words, there exists $b(\xi) \in \F$ such that (\ref{Eq2.25}) holds. In the noncommutative case, the relation between $(\alpha, \beta , \gamma)$-positive functions and elements of the form $\overline{b (z)} \star_{\alpha} \star_{\beta} \star_{\gamma} b(z)$ is less clear. To begin with, if we consider a function $f(z) = \overline{b (z)} \star_{\hbar} \star_{\theta} \star_{\eta} b(z)$, then eq.(\ref{Eq4.2.9}) reveals that $f$ is not of the $(\hbar, \theta , \eta)$-positive type but rather of the $(\tilde{\hbar}, \tilde{\theta} , \tilde{\eta})$-positive type. Moreover, there is an additional complication. If we look carefully at the noncommutative symplectic Fourier transform (\ref{Eq4.2.1}) (and contrary to what happens with the commutative symplectic Fourier transform (\ref{Eq2.19})), it depends on the parameters $\hbar, \theta, \eta$, even if we are considering functions of $(\alpha, \beta, \gamma)$-positive type, with $(\alpha, \beta, \gamma) \ne (\tilde{\hbar}, \tilde{\theta}, \tilde{\eta})$. To circumvent this difficulty, let us consider a function $f \in \F$ of the form:
\begin{equation}
f(z) = \overline{b(z)} \star_{\alpha} \star_{\beta} \star_{\gamma} b(z),
\label{Eq4.2.20}
\end{equation}
with
\begin{equation}
\int_{\bkR^4} |b(z)|^2 dz=1,
\label{Eq4.2.21}
\end{equation}
where $\alpha, \beta, \gamma$ have the dimensions of $\hbar, \theta, \eta$, respectively and where:
\begin{equation}
\alpha^2 \ne \beta \gamma.
\label{Eq4.2.22}
\end{equation}
We may thus construct a new complete Weyl-Wigner formulation by defining:
\begin{equation}
\alpha {\bf \Omega}' = \left(
\begin{array}{c c}
\beta {\bf E} & \alpha {\bf I}_{2 \times 2}\\
- \alpha {\bf I}_{2 \times 2} & \gamma {\bf E}
\end{array}
\right).
\label{Eq4.2.23}
\end{equation}
We also define appropriate D maps ${\bf S}'$:
\begin{equation}
{\bf S}' {\bf J} {\bf S}'^T = {\bf \Omega}', \hspace{1 cm} \det {\bf S}' = |1 - \zeta'|, \hspace{1 cm} \zeta'=\frac{\beta \gamma}{\alpha}.
\label{Eq4.2.24}
\end{equation}
Our whole construction of noncommutative quantum mechanics in phase-space \cite{Bastos1} goes through with $(\hbar, \theta, \eta)$ replaced by $(\alpha, \beta, \gamma)$. Consequently if $f$ is of the form (\ref{Eq4.2.20},\ref{Eq4.2.21}), there must exist a normalized function $g(\xi)$ of $\alpha$-positive type such that:
\begin{equation}
f(z) = \frac{1}{\left| Pf({\bf \Omega'}) \right|} g ({\bf S}'^{-1} z).
\label{Eq4.2.25}
\end{equation}
The corresponding noncommutative symplectic Fourier transform is:
\begin{equation}
\begin{array}{c}
\tilde f^{{\bf \Omega}} (a) = \int_{\bkR^4} f (z) \exp \left( i a^T {\bf \Omega}^{-1} z \right) ~dz=
\frac{1}{\left| Pf({\bf \Omega'}) \right|} \int_{\bkR^4} g ({\bf S}'^{-1} z) \exp \left( i a^T {\bf \Omega}^{-1} z \right) ~dz= \\
\\
= \int_{\bkR^4}  g (\xi) \exp \left( i a^T {\bf \Omega}^{-1} {\bf S}' {\bf J}^T {\bf J} \xi \right) ~d \xi= \tilde g^{{\bf J}} ({\bf J} {\bf S}'^T {\bf \Omega}^{-1} a).
\end{array}
\label{Eq4.2.26}
\end{equation}
Since $\tilde g^{{\bf J}}$ is of $\alpha$-positive type, for any positive integer $m$ and any set of points $b_1, \cdots, b_m$, the matrices
\begin{equation}
M_{jk} = \tilde g^{{\bf J}} (b_j - b_k) \exp \left( - \frac{i \alpha}{2} b_k^T {\bf J} b_j \right),
\label{Eq4.2.27}
\end{equation}
are hermitian and non-negative. If we define $b_i = {\bf J}{\bf S}'^T {\bf \Omega}^{-1} a_i$, $(i=1, \cdots, m)$, we get from (\ref{Eq4.2.26}):
\begin{equation}
M_{jk} = \tilde f^{{\bf \Omega}} (a_j - a_k) \exp \left(  \frac{i \alpha}{2} a_k^T {\bf \Omega}^{-1} {\bf S}'{\bf J}{\bf S}'^T {\bf \Omega}^{-1} a_j \right),
\label{Eq4.2.28}
\end{equation}
From (\ref{Eq4.2.24}) it then follows that:
\begin{equation}
M_{jk} = \tilde f^{{\bf \Omega}} (a_j - a_k) \exp \left(  \frac{i \alpha}{2} a_k^T {\bf \Omega}^{-1} {\bf \Omega}' {\bf \Omega}^{-1} a_j \right).
\label{Eq4.2.29}
\end{equation}
Now notice that
\begin{equation}
\alpha  {\bf \Omega}^{-1} {\bf \Omega}' {\bf \Omega}^{-1} = {\bf \Lambda} (\alpha^{\sharp}, \beta^{\sharp}, \gamma^{\sharp}),
\label{Eq4.2.30}
\end{equation}
where:
\begin{equation}
\alpha^{\sharp} = \frac{\alpha \tilde{\hbar} (1 + \zeta) - \tilde{\eta} \beta - \tilde{\theta} \gamma}{\hbar(1- \zeta)} , \hspace{1 cm}
\beta^{\sharp} = \frac{2 \alpha \tilde{\hbar} \tilde{\theta} - \tilde{\hbar}^2 \beta - \tilde{\theta}^2 \gamma}{\hbar^2} , \hspace{1 cm} \gamma^{\sharp} = \frac{2 \alpha \tilde{\hbar} \tilde{\eta} - \tilde{\hbar}^2 \gamma - \tilde{\eta}^2 \beta}{\hbar^2}
\label{Eq4.2.31}
\end{equation}
This system is easily inverted:
\begin{equation}
\alpha = (1 + \zeta) \alpha^{\sharp} - \frac{1}{\hbar} (\eta \beta^{\sharp} + \theta \gamma^{\sharp}), \hspace{1 cm} \beta = \frac{2 \hbar \zeta \alpha^{\sharp} - \eta \beta^{\sharp} - \theta \zeta \gamma^{\sharp}}{\eta}, \hspace{1 cm} \gamma = \frac{2 \hbar \zeta \alpha^{\sharp} - \theta \gamma^{\sharp} - \zeta \eta \beta^{\sharp}}{\theta}.
\label{Eq4.2.32}
\end{equation}
Consequently, for $\alpha, \beta, \gamma$ such that (\ref{Eq4.2.22}) holds, we have:

\begin{theorem} If $f(z) \in \F$ is such that (\ref{Eq4.2.20},\ref{Eq4.2.21}) holds for some $b(z) \in \F$, then $f$ is such that $\tilde f^{{\bf \Omega}} (0)=1$ and $\tilde f^{{\bf \Omega}}$ is continuous and of the $(\alpha^{\sharp}, \beta^{\sharp}, \gamma^{\sharp} )$-positive type, with $\alpha^{\sharp}, \beta^{\sharp}, \gamma^{\sharp}$ given by (\ref{Eq4.2.31}). Conversely, if $\tilde f^{{\bf \Omega}} (0)=1$ and $\tilde f^{{\bf \Omega}}$ is continuous and of the $(\alpha^{\sharp}, \beta^{\sharp}, \gamma^{\sharp} )$-positive type, then there exists $b(z) \in \F$ such that (\ref{Eq4.2.20},\ref{Eq4.2.21}) hold, for $\alpha, \beta, \gamma$, given by (\ref{Eq4.2.32}) as long as $\alpha^2 \ne \beta \gamma$.
\end{theorem}

\begin{corollary} If $(\alpha^{\sharp}, \beta^{\sharp}, \gamma^{\sharp} ) \in \W^{NC} (f)$, then $(- \alpha^{\sharp}, - \beta^{\sharp}, - \gamma^{\sharp} ) \in \W^{NC} (f)$.
\end{corollary}

\begin{proof} Again this is a consequence of (i) the fact that $a(z)  \star_{\alpha} \star_{\beta} \star_{\gamma} b(z) = b (z)  \star_{- \alpha} \star_{- \beta} \star_{- \gamma} a(z)$, (ii) Theorem 4.11 and (iii) the fact that under the replacement $(\alpha, \beta, \gamma) \to (- \alpha,- \beta, - \gamma)$ in (\ref{Eq4.2.31}), we get $(\alpha^{\sharp}, \beta^{\sharp}, \gamma^{\sharp}) \to (- \alpha^{\sharp},- \beta^{\sharp}, - \gamma^{\sharp})$.
\end{proof}

\begin{remark} If we set $\beta = \gamma =0$ and $\alpha = \hbar$ in (\ref{Eq4.2.31}), we obtain $(\alpha^{\sharp}, \beta^{\sharp}, \gamma^{\sharp} ) = (1- \zeta)^{-2} ( \hbar (1 + \zeta) , 2 \theta , 2 \eta )$. Consequently, if $f$ is a Wigner measure, then:
\begin{equation}
\frac{( \tilde{\hbar} (1 + \zeta) , 2 \tilde{\theta} ,  2 \tilde{\eta} )}{1 - \zeta} \in \W^{NC} (f).
\label{Eq4.2.33}
\end{equation}
Obviously, from Bochner's theorem, if $(0,0,0) \in \W^{NC} (f)$, then $f$ is everywhere non-negative.

Theorems 4.10 and 4.11, then suggest a way of constructing functions which are simultaneously commutative and noncommutative Wigner measures. Indeed, let $f,g \in \F$ be such that:

\vspace{0.3 cm}
\noindent
(i) $f(z)$ is a NCWM, i.e. (cf. Theorem 4.8 and Corollary 4.12):
\begin{equation}
\left\{(\tilde{\hbar} , \tilde{\theta}, \tilde{\eta} ), (- \tilde{\hbar} , - \tilde{\theta}, - \tilde{\eta} ) \right\} \subseteq \W^{NC} (f)
\label{Eq4.2.34}
\end{equation}

\vspace{0.3 cm}
\noindent
(ii) $g(z)$ is such that:
\begin{equation}
\left\{(0,0,0), (\alpha^{\sharp} , \beta^{\sharp}, \gamma^{\sharp} ), (- \alpha^{\sharp} , - \beta^{\sharp}, - \gamma^{\sharp} ) \right\} \subseteq \W^{NC} (g)
\label{Eq4.2.35}
\end{equation}
where $\alpha^{\sharp} , \beta^{\sharp}, \gamma^{\sharp} $ are given by:
\begin{equation}
(\alpha^{\sharp} , \beta^{\sharp}, \gamma^{\sharp} )= \frac{(2 \tilde{\hbar} \zeta , (1 + \zeta) \tilde{\theta} , (1 + \zeta) \tilde{\eta} )}{1 - \zeta}.
\label{Eq4.2.36}
\end{equation}
From Theorem 4.10, we conclude that the NCNW spectrum of the convolution $f \natural g$ contains the elements:
\begin{equation}
(\tilde{\hbar} , \tilde{\theta}, \tilde{\eta} ) + (\alpha^{\sharp} , \beta^{\sharp}, \gamma^{\sharp} ) = \frac{( \tilde{\hbar} (1 + \zeta) , 2 \tilde{\theta} ,  2 \tilde{\eta} )}{1 - \zeta} , \hspace{1 cm} (\tilde{\hbar} , \tilde{\theta}, \tilde{\eta} ) + (0,0,0) = (\tilde{\hbar} , \tilde{\theta}, \tilde{\eta} )
\label{Eq4.2.37}
\end{equation}
The first element means that the convolution is a Wigner measure (\ref{Eq4.2.33}), whereas the second one entails that it is equally a NCWM (\ref{Eq4.2.34}). Functions $g$ satisfying (\ref{Eq4.2.35}) are easy to construct. Indeed from (\ref{Eq4.2.20},\ref{Eq4.2.31}), we conclude that any function of the form
\begin{equation}
g(z) = \overline{b(z)} \star_{\theta} \star_{\eta} b(z), \hspace{1 cm} b(z) \in \F
\label{Eq4.2.38}
\end{equation}
contains $\pm (\alpha^{\sharp} , \beta^{\sharp}, \gamma^{\sharp} )$ (\ref{Eq4.2.36}) in its NCNW spectrum. Moreover, it is easy to check that if $b(z)$ is a Gaussian, then $g(z)$ in (\ref{Eq4.2.38}) is positive, as it is another Gaussian. Notice that we can safely set $\alpha=0$. Indeed it always appears in all the formulae in the combination $\alpha {\bf \Omega}'$ (\ref{Eq4.2.23}) which is regular as $\alpha \downarrow 0$.
\end{remark}

\end{subsection}

\begin{subsection}{Constructing functions in $\F^C$, $\F^{NC}$ and $\L$}

Our purpose now is to investigate how the sets of Wigner measures ($\F^C$), noncommutative Wigner measures ($\F^{NC}$) and Liouville measures ($\L$) relate to each other. The latter is the set of real, normalized phase-space functions, which are everywhere non-negative, i.e. functions whose symplectic Fourier transform is of $0$-positive type or whose noncommutative symplectic Fourier transform is of $(0,0,0)$-positive type. Let us then define the sets:
\begin{equation}
\begin{array}{l l l}
\Omega_1 = \F^C \backslash ( \F^{NC} \cup \L ), & \Omega_2 = \F^{NC} \backslash ( \F^C \cup \L ) , & \Omega_3 = \L \backslash (\F^C \cup \F^{NC} ), \\
& & \\
\Omega_4 = (\F^C \cap  \F^{NC}) \backslash \L,  & \Omega_5 = (\F^C \cap \L ) \backslash \F^{NC} , & \Omega_6 = (\F^{NC} \cap \L ) \backslash \F^C,  \\
& & \\
\Omega_7 =  \F^C \cap \F^{NC} \cap \L.
\end{array}
 \label{Eq4.3.1}
\end{equation}
The remainder of this section is devoted to proving Lemma 4.14. We depicted the content of the lemma in Figure 1.

\begin{figure}
\begin{center}
\includegraphics[scale=0.5]{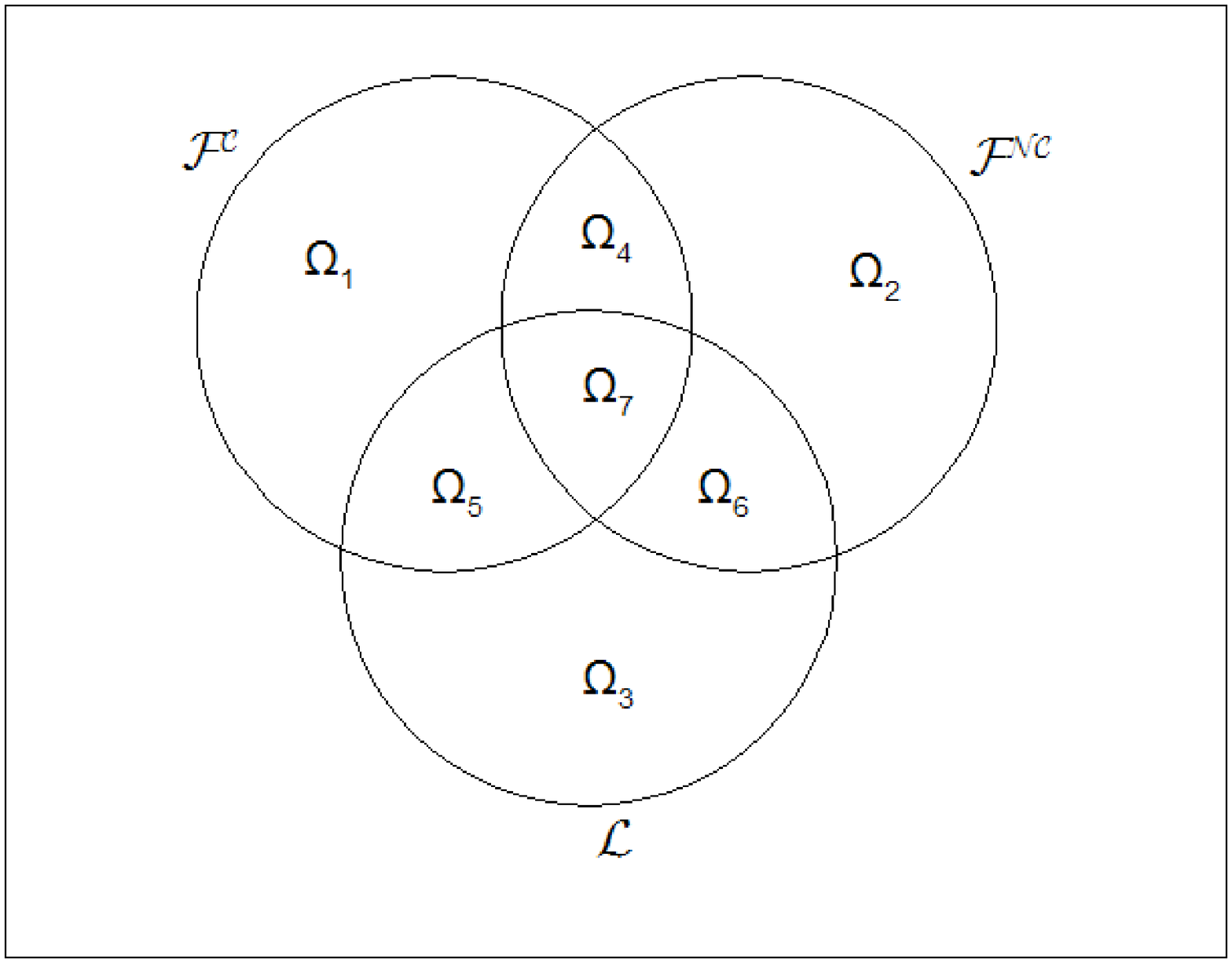}
\caption{Different sets of functions and their intersection.}
\end{center}
\label{diagrama}
\end{figure}

\begin{lemma} The sets $\Omega_i$ $(i=1, \cdots, 7)$ are all non-empty.
\end{lemma}

\begin{proof} To prove the lemma we shall construct explicitly families of functions in each of the sets, by resorting to the properties of functions in $\F^C$, $\F^{NC}$ and $\L$. Let us start with the simplest case:

\vspace{0.3 cm}
\noindent
{\bf A function in $\Omega_3$:} The function
\begin{equation}
f_3 (q,p) = \frac{1}{\pi^2 ab} \exp \left( - \frac{q^2}{a} - \frac{p^2}{b} \right), \hspace{1 cm} a, b >0, ~ab < \hbar^2 (1 - \zeta)
\label{Eq4.3.2}
\end{equation}
belongs to $\Omega_3$.

\vspace{0.3 cm}
\noindent
It is obvious that $f_3 \in \L$, since it is real, normalized and everywhere positive. To prove that it does not belong to $\F^C \cup \F^{NC}$, let us compute its purity:
\begin{equation}
\int_{\bkR^2} \int_{\bkR^2} \left[f_3 (q,p) \right]^2 ~dq dp= \frac{1}{(2 \pi)^2 ab} > \frac{1}{(2 \pi \hbar \sqrt{1- \zeta})^2 } > \frac{1}{(2 \pi \hbar)^2 }.
\label{Eq4.3.3}
\end{equation}
Consequently, from (\ref{Eq2.15},\ref{Eq4.3}), we conclude that $f_3 \in \Omega_3$.

\vspace{0.3 cm}
\noindent
{\bf A function in $\Omega_5$:} The function
\begin{equation}
f_5 (q,p) = \frac{1}{(\pi \hbar)^2} \exp \left( - \frac{q^2}{a} - \frac{ap^2}{\hbar^2} \right), \hspace{1 cm} 0 <a < \theta
\label{Eq4.3.4}
\end{equation}
belongs to $\Omega_5$.

\vspace{0.3 cm}
\noindent
Let us consider the normalized wave function $\psi_5 (q) = \frac{1}{\sqrt{\pi a}} \exp \left(- \frac{q^2}{2a} \right)$. A simple calculation shows that $f_5$ is the Wigner measure associated with $\psi_5$. By construction, we conclude that $f_5 \in \F^C$. Moreover $f_5 $ is positive, which means that $f_5 \in \L$. It remains to prove that $f_5 \notin \F^{NC}$. Let us compute its $\theta$-purity (\ref{Eq4.1.22}). The marginal distribution reads
\begin{equation}
\P_q (q) = \int_{\bkR^2} f_5 (q,p) ~dp= \frac{1}{\pi a} \exp \left(- \frac{q^2}{a} \right).
\label{Eq4.3.5}
\end{equation}
And thus:
\begin{equation}
\int_{\bkR^2} \left[\P_q (q) \right]^2 ~dq= \frac{1}{2 \pi a} >   \frac{1}{2 \pi \theta}.
\label{Eq4.3.6}
\end{equation}
We conclude that $f_5 \notin \F^{NC}$.

\vspace{0.3 cm}
\noindent
{\bf A function in $\Omega_1$:} The function
\begin{equation}
f_1 (q,p) = \frac{8}{3a(\pi \hbar)^2}\left( q_1^2 + \frac{9 a^2 p_1^2}{16 \hbar^2} - \frac{3a}{8} \right) \exp \left( - \frac{4q^2}{3a} - \frac{3ap^2}{4\hbar^2} \right), \hspace{1 cm} 0 <a < \theta
\label{Eq4.3.7}
\end{equation}
belongs to $\Omega_1$.

\vspace{0.3 cm}
\noindent
Let us consider the normalized wave function $\psi_1 (q) = \frac{4}{3a}\sqrt{\frac{2}{\pi}} q_1  \exp \left(- \frac{2q^2}{3a} \right)$. The corresponding Wigner measure is $f_1 (q,p)$. However, $f_1$ is negative inside the ellipse $q_1^2 + \left( \frac{3 a p_1}{4 \hbar} \right)^2 < \frac{3 a}{8}$. And thus $f_1 \in \F^C \backslash \L$. Finally, let us prove that $f_1 \notin \F^{NC}$. Integration over the momenta yields:
\begin{equation}
\P_q (q) = \int_{\bkR^2}  f_1 (q,p) ~dp= \frac{32}{9 \pi a^2} q_1^2 \exp \left(- \frac{4q^2}{3a} \right).
\label{Eq4.3.8}
\end{equation}
And thus:
\begin{equation}
\int_{\bkR^2} \left[\P_q (q) \right]^2 ~dq= \frac{1}{2 \pi a} >   \frac{1}{2 \pi \theta},
\label{Eq4.3.9}
\end{equation}
which means that $f_1 \notin \F^{NC}$.

\vspace{0.3 cm}
\noindent
{\bf A function in $\Omega_6$:} The function
\begin{equation}
\begin{array}{c}
f_6 (q,p) = \frac{1}{(\pi \hbar \sqrt{1- \zeta})^2}  \exp \left\{ - \left(\frac{2 -\zeta}{1- \zeta} \right) \left(2 a q_1^2 + \frac{q_2^2}{2a \theta^2}\right) \right.\\
\\
 \left. - \left( \frac{\theta}{\hbar \sqrt{1 - \zeta}} \right)^2 \left( \frac{p_1^2}{2 a\theta^2} + 2 a p_2^2 \right) - \frac{2 \theta}{\hbar (1 - \zeta)} \left( 2 a q_1 p_2 - \frac{q_2 p_1}{2 a \theta^2} \right) \right\} , \hspace{1 cm} a >0
\label{Eq4.3.10}
\end{array}
\end{equation}
belongs to $\Omega_6$.

\vspace{0.3 cm}
\noindent
It is easy to check that $f_6$ is a NCWM of the form (\ref{Eq4.1.24}) with $\phi_1 (q_1)= \phi_2 (q_1) = \left( \frac{2a}{\pi} \right)^{\frac{1}{4}} \exp( - a q_1^2)$. Moreover, we already know from Lemma 4.5 that states of the form (\ref{Eq4.1.24}) cannot be Wigner measures. Finally, since $f_6$ is everywhere positive, we conclude that $f_6 \in \Omega_6$.

\vspace{0.3 cm}
\noindent
{\bf A function in $\Omega_2$:} The function
\begin{equation}
\begin{array}{c}
f_2 (q,p) = \frac{4a}{(\pi \hbar \sqrt{1- \zeta})^2} \left(q_1^2 + \frac{q_2^2}{4 a^2 \theta^2} - \frac{1}{4 a} \right)  \exp \left\{ - \left(\frac{2 -\zeta}{1- \zeta} \right) \left(2 a q_1^2 + \frac{q_2^2}{2a \theta^2}\right) \right.\\
\\
 \left. - \left( \frac{\theta}{\hbar \sqrt{1 - \zeta}} \right)^2 \left( \frac{p_1^2}{2 a\theta^2} + 2 a p_2^2 \right) - \frac{2 \theta}{\hbar (1 - \zeta)} \left( 2 a q_1 p_2 - \frac{q_2 p_1}{2 a \theta^2} \right) \right\} , \hspace{1 cm} a >0
\label{Eq4.3.11}
\end{array}
\end{equation}
belongs to $\Omega_2$.

\vspace{0.3 cm}
\noindent
If we choose $\phi_1 (q_1)=  \left( \frac{2a}{\pi} \right)^{\frac{1}{4}} \exp( - a q_1^2)$ and $\phi_2 (q_1) = \left( \frac{32 a^3}{\pi} \right)^{\frac{1}{4}} q_1 \exp( - a q_1^2)$ and substitute in (\ref{Eq4.1.24}), we obtain $f_2$. From Lemma 4.5 we conclude that $f_2 \in \F^{NC} \backslash \F^C$. Since $f_2$ is negative inside the ellipse $q_1^2 + \left( \frac{q_2}{2 a \theta} \right)^2 < \frac{1}{4 a}$, we conclude that $f_2 \notin \L$.

\vspace{0.3 cm}
\noindent
{\bf A function in $\Omega_7$:} Any function of the form $f_7 (z) = (f_6 \natural g) (z)$ with $f_6$ an element of $\Omega_6$ and
\begin{equation}
\begin{array}{c}
g (q,p) = \frac{1}{\pi^2 cd} \exp \left\{ - \frac{q^2}{c}  - \frac{p^2}{d} \right), \hspace{1 cm} c \ge \theta, ~ d \ge \eta
\label{Eq4.3.12}
\end{array}
\end{equation}
belongs to $\Omega_7$.

\vspace{0.3 cm}
\noindent
Let us choose $\alpha, \beta >0$ such that:
\begin{equation}
c = \frac{1 + \alpha^2 \theta^2}{2 \alpha} , \hspace{1 cm} d = \frac{1 + \beta^2 \eta^2}{2 \beta}.
\label{Eq4.3.13}
\end{equation}
With this choice, $c,d$ automatically satisfy $c \ge \theta$ and $d \ge \eta$. Moreover, we define:
\begin{equation}
b (q,p) = \frac{2}{\pi} \sqrt{\alpha \beta} \exp \left( - \alpha q^2 - \beta p^2 \right).
\label{Eq4.3.14}
\end{equation}
Using the kernel representations (\ref{Eq3.15}) of the star-products it is straightforward to show that:
\begin{equation}
g (z) = b(z) \star_{\theta} \star_{\eta} b(z).
\label{Eq4.3.15}
\end{equation}
From our discussion in Remark 4.13, we know that under these circumstances, the convolution of $f_6$ and $g$ is simultaneously a Wigner measure and a NCWM. Finally, since $g(z)$ is positive, its convolution with another positive function is again positive.

\vspace{0.3 cm}
\noindent
{\bf A function in $\Omega_4$:} The function
\begin{equation}
f_4 (q,p) = \frac{1}{3(\pi \hbar)^2} \left[ \frac{2}{3 \theta} \left( q - \frac{\theta}{\hbar} {\bf E} p \right)^2 -1 \right] \exp \left( - \frac{2 q^2}{3 \theta} - \frac{2 \theta p^2}{3 \hbar^2} - \frac{2 }{3 \hbar} q \cdot {\bf E} p \right),
\label{Eq4.3.16}
\end{equation}
is a function of $\Omega_4$.

\vspace{0.3 cm}
\noindent
Let us choose $g(z)$ of the form (\ref{Eq4.3.12}) with $c= \theta$ and $d= \eta$. Moreover, let us consider the function $f_2 (z)$ in (\ref{Eq4.3.11}) with $a= \frac{1}{2 \theta}$. Then it is easy to show that $f_4 (z) = (g \natural f_2) (z)$. From Remark 4.13, we conclude that $f_4 \in \F^C \cup \F^{NC}$. However, $f_4$ is negative for $\left(q- \frac{\theta}{\hbar} {\bf E} p \right)^2 < \frac{3 \theta}{2}$, which means that $f_4 \notin \L$. This completes the proof of the lemma.
\end{proof}

\begin{remark} The function $f_6$ in (\ref{Eq4.3.10}) reveals that functions of the form (\ref{Eq4.1.24}) saturate the $\theta$-purity but not the $\eta$-purity. Indeed by a simple calculation, we obtain:
\begin{equation}
\P_p (p) = \int_{\bkR^2}  f_6 (q,p) ~dq= \frac{\theta}{\pi \hbar^2 (2 - \zeta)} \exp \left( - \frac{1}{2 a \hbar^2 (2 - \zeta)} ( p_1^2 + 4 a^2 \theta^2 p_2^2 ) \right).
\label{Eq4.3.17}
\end{equation}
Consequently
\begin{equation}
\int_{\bkR^2}  \left[\P_p (p) \right]^2 ~dp= \frac{\zeta}{2\pi \eta (2 - \zeta)}.
\label{Eq4.3.18}
\end{equation}
Since $\zeta <1$, we conclude that this is strictly smaller than $\frac{1}{2\pi \eta }$. By the same token, we can show that states of the form (\ref{Eq4.1.25}), albeit maximizing the $\eta$-purity, need not saturate the $\theta$-purity.
\end{remark}

\end{subsection}

\end{section}

\begin{section}{Appendix}

In this appendix we prove that for a D transformation, the associated matrix ${\bf S} \in \D_{{\bf \Omega}} (2d; \bkR)$ satisfies:
\begin{equation}
\det {\bf S} = \sqrt{\det {\bf \Omega}} = | Pf ({\bf \Omega})| >0.
\label{Eqa.1}
\end{equation}
To prove this, we first derive the following Lemma:
\begin{lemma} Under the assumption (\ref{Eq3.2}), the sign of the Pfaffian of the matrix ${\bf \Omega}$ reads:
\begin{equation}
sign \left(Pf ({\bf \Omega}) \right) = (-1)^{d(d-1)/2}.
\label{Eqa.2}
\end{equation}
\end{lemma}

\begin{proof}
Let $(\omega_{\alpha \beta} )$ $(\alpha,\beta=1, \cdots, 2d)$ denote the elements of ${\bf \Omega}$. The Pfaffian of ${\bf \Omega}$ can be obtained from the following recursive formula \cite{Pfaffian}:
\begin{equation}
Pf ( {\bf \Omega} ) = \sum_{\alpha=2}^{2d} (-1)^{\alpha} \omega_{1, \alpha} Pf   ( {\bf \Omega}_{\hat 1, \hat{\alpha}} ),
\label{Eqa.3}
\end{equation}
where ${\bf \Omega}_{\hat 1, \hat{\alpha}}$ denotes the matrix ${\bf \Omega}$ with both the 1st and $\alpha$-th rows and columns removed. From (\ref{Eq3.1}), we get:
\begin{equation}
Pf ( {\bf \Omega} ) = \sum_{i=2}^d (-1)^i \frac{\theta_{1i}}{\hbar} Pf   ( {\bf \Omega}_{\hat 1 ,\hat i} ) + (-1)^{d+1} Pf   ( {\bf \Omega}_{\hat 1 ,\hat{d+1}} )  . \label{Eqa.4}
\end{equation}
A term which is independent of the elements of ${\bf \Theta}$ and ${\bf N}$ can only be found in $(-1)^{d+1} Pf   ( {\bf \Omega}_{\hat 1, \hat{d+1}} )$. Suppose that $d \ge 3$. If we apply the recursive formula (\ref{Eqa.3}) again we obtain a term of the form $(-1)^{d+1} (-1)^d Pf ({\bf A}_2 )$ where ${\bf A}_2$ is obtained from ${\bf \Omega}$ by removing the 1st, 2nd, (d+1)th and (d+2)th rows and columns. After i steps we obtain a term  $(-1)^{d+1} (-1)^d \cdots (-1)^{d+2-i}  Pf ({\bf A}_i )$ where ${\bf A}_i$ is obtained from ${\bf \Omega}$ by removing the 1st, 2nd,..., ith, and (d+1)th, (d+2)th,..., (d+i)th rows and columns. We terminate this process when $i=d-2$. We thus obtain:
\begin{equation}
\begin{array}{c}
(-1)^{d+1} (-1)^d \cdots (-1)^4  Pf \left(
\begin{array}{c c c c}
0 & \frac{\theta_{d-1,d}}{\hbar} & 1 & 0\\
\frac{\theta_{d,d-1}}{\hbar} & 0 & 0 & 1\\
-1 & 0 & 0 & \frac{\eta_{d-1,d}}{\hbar}\\
0 & -1 & \frac{\eta_{d,d-1}}{\hbar} & 0
\end{array}
\right) = \\
\\
= \left(\frac{\theta_{d-1,d} \eta_{d-1,d}}{\hbar^2} -1 \right) (-1)^{\sum_{i=4}^{d+1} i}.
\end{array}
\label{Eqa.5}
\end{equation}
And thus the term independent of the elements of ${\bf \Theta}$ and ${\bf N}$ is $(-1)^{d(d-1) /2}$. We leave to the reader the simple task of verifying that this result also holds when $d=2$.

Let us now turn to the $\theta$ and $\eta$ dependent terms. We resort to the definition of the Pfaffian \cite{Pfaffian}:
\begin{equation}
Pf ({\bf \Omega}) = \frac{1}{2^d d!} \sum_{\sigma \in S_{2d}} sgn (\sigma) \Pi_{i=1}^d \omega_{\sigma (2i-1), \sigma (2i)},
\label{Eqa.6}
\end{equation}
where $S_{2d}$ is the symmetric group and $sgn(\sigma)$ is the signature of the permutation $\sigma$. Moreover, we use the following notation. If $d=2$, for instance, then we consider the permutations of the set $\left\{1,2,3,4 \right\}$. As an example, consider $\sigma = \left\{3,1,4,2 \right\}$. Then we write: $\sigma(1) =3$, $\sigma(2) =1$, $\sigma(3) =4$, and $\sigma(4) =2$.

Suppose that in the string $\Pi_{i=1}^d \omega_{\sigma (2i-1), \sigma (2i)}$ we pick $k$ elements of the matrix $\hbar^{-1} {\bf \Theta}$, $p$ elements of the matrix $\hbar^{-1} {\bf N}$ and $l$ elements of the matrix ${\bf I}$ or $- {\bf I}$. Then, of course:
\begin{equation}
k +l +p =d.
\label{Eqa.7}
\end{equation}
If we pick $l$ elements from ${\bf I}$ or $- {\bf I}$, then the remaining $k+p$ terms can only be taken from ${\bf \Omega}$ when $2l$ lines and rows have been eliminated. In particular, we remove $l$ lines and rows from $\hbar^{-1} {\bf \Theta}$. That leaves us with $(d-l-1)(d-l)/2$ non-vanishing independent parameters in $\hbar^{-1} {\bf \Theta}$. Each time we choose one of the latter for our string $\Pi_{i=1}^d \omega_{\sigma (2i-1), \sigma (2i)}$, we have to eliminate another 2 lines and 2 columns. So if we pick k elements out of the $(d-l-1)(d-l)/2$ non-vanishing independent elements of $\hbar^{-1} {\bf \Theta}$, we remove $2k$ lines and columns. We are left with $(d-l-2k-1)(d-l-2k)/2$ non-vanishing independent elements. But this is only possible if:
\begin{equation}
2k \le d-l.
\label{Eqa.8}
\end{equation}
A similar argument leads to:
\begin{equation}
2p \le d-l.
\label{Eqa.9}
\end{equation}
Now, (\ref{Eqa.8}) and (\ref{Eqa.9}) are only compatible with (\ref{Eqa.7}) if:
\begin{equation}
k=p= \frac{ d-l}{2}.
\label{Eqa.10}
\end{equation}
This means that in each string we have exactly the same number of elements of $\hbar^{-1} {\bf \Theta}$ and $\hbar^{-1} {\bf N}$. This proves that:
\begin{equation}
Pf ({\bf \Omega}) = (-1)^{d(d-1)/2} + P_{\left[d/2 \right]} ,
\label{Eqa.11}
\end{equation}
where $ P_{\left[d/2 \right]} $ is a homogeneous polynomial of degree $\left[d/2 \right]$ (the integral part of $d/2$) in the dimensionless variables $\theta_{ij} \eta_{kl} / \hbar^2$ with $1 \le i <j \le d$ and $1 \le k <l \le d$.

Let us define:
\begin{equation}
\zeta = {\mbox max} \left\{\theta_{ij} \eta_{kl} / \hbar^2,~~ 1 \le i <j \le d, ~ 1 \le k <l \le d\right\}.
\label{Eqa.12}
\end{equation}
Let $\sigma'$ be the permutation which yields the contribution $(-1)^{d(d-1)/2}$ to the Pfaffian and let $S_{2d}': = S_{2d} \backslash \left\{ \sigma' \right\}$. We thus have:
\begin{equation}
\left| P_{\left[d/2 \right]} \right| =
\left|\frac{1}{2^d d!} \sum_{\sigma \in S_{2d}'} sgn (\sigma) \Pi_{i=1}^d \omega_{\sigma (2i-1), \sigma (2i)} \right| \le  \frac{1}{2^d d!}\sum_{\sigma \in S_{2d}'} \Pi_{i=1}^d \left| \omega_{\sigma (2i-1), \sigma (2i)} \right|.
\label{Eqa.13}
\end{equation}
If a string $\Pi_{i=1}^d  \omega_{\sigma (2i-1), \sigma (2i)} $ contains $k$ elements of $\hbar^{-1} {\bf \Theta}$ and $k$ elements of $\hbar^{-1} {\bf N}$, then
\begin{equation}
 \Pi_{i=1}^d \left| \omega_{\sigma (2i-1), \sigma (2i)} \right| \le \zeta^k < \zeta,
 \label{Eqa.14}
 \end{equation}
 where we used $\zeta <1$. Since there are $d! -1 < d!$ elements in $S_{2d}'$, we conclude that:
\begin{equation}
\left| P_{\left[d/2 \right]} \right| < \frac{\zeta }{2^d} <1.
\label{Eqa.15}
\end{equation}
This yields the desired result.
\end{proof}
An immediate consequence of this Lemma is that the matrix ${\bf \Omega}$ is invertible and the skew-symmetric form (\ref{Eq3.3}) is non-degenerate as advertised.

\begin{proposition}
A matrix ${\bf S}$ associated with a D transformation has positive determinant.
\end{proposition}

\begin{proof}
We use the well known formula \cite{Pfaffian}:
\begin{equation}
Pf \left({\bf B} {\bf A} {\bf B}^T \right) = \det ({\bf B}) Pf ({\bf A}),
\label{Eqa.16}
\end{equation}
which holds for any $2d \times 2d$ skew-symmetric matrix ${\bf A}$ and any $2d \times 2d$ matrix ${\bf B}$. If we apply this formula to (\ref{Eq3.5}), we obtain:
\begin{equation}
\det ({\bf S}) Pf ({\bf J}) = Pf ({\bf \Omega}).
\label{Eqa.17}
\end{equation}
For an arbitrary $d \times d$ matrix ${\bf M}$ we have \cite{Pfaffian}:
\begin{equation}
Pf \left(
\begin{array}{c c}
0 & {\bf M}\\
- {\bf M}^T & 0
\end{array}
\right) = (-1)^{d (d-1)/2} \det {\bf M}.
\label{Eqa.18}
\end{equation}
We conclude that:
\begin{equation}
Pf( {\bf J})=  (-1)^{d (d-1)/2}.
\label{Eqa.19}
\end{equation}
From (\ref{Eqa.17},\ref{Eqa.19}) and Lemma 6.1, we infer that:
\begin{equation}
\det ( {\bf S}) >0.
\label{Eqa.20}
\end{equation}
\end{proof}

\end{section}

\subsection*{Acknowledgments}

\vspace{0.3cm}

\noindent
The authors would like to thank O. Bertolami for useful discussions and for reading the manuscript. The work of CB is supported by Funda\c{c}\~{a}o para a Ci\^{e}ncia e a Tecnologia (FCT)
under the fellowship SFRH/BD/24058/2005. The work of NCD and JNP was partially supported by the grant PTDC/MAT/69635/2006 of the FCT.


\vspace{0.3cm}

\begin{thebibliography}{99}

\bibitem{Antoine1} J.P. Antoine: Dirac formalism and symmetry problems in quantum mechanics. I: General Dirac formalism, {\it J. Math. Phys.} {\bf 710} (1969) 53.

\bibitem{Antoine2} J.P. Antoine: Dirac formalism and symmetry problems in quantum mechanics. II: Symmetry problems, {\it J. Math. Phys.} {\bf 710} (1969) 2276.

\bibitem{Bastos} C. Bastos, O. Bertolami: Berry phase in the gravitational quantum well and the Seiberg-Witten map, {\it Phys. Lett. A} {\bf 372} (2008) 5556.

\bibitem{Bastos1} C. Bastos, O. Bertolami, N.C. Dias, J.N. Prata: Weyl-Wigner formulation of noncommutative quantum mechanics, {\it J. Math. Phys.} {\bf 49} (2008) 072101.

\bibitem{Bastos2} C. Bastos, O. Bertolami, N.C. Dias, J.N. Prata: Phase-space noncommutative quantum cosmology, {\it Phys. Rev. D} {\bf 78} (2008) 023516.

\bibitem{Bayen1} F. Bayen, M. Flato, C. Fronsdal, A. Lichnerowicz, D. Sternheimer: {\it Deformation theory and quantization I. Deformations of symplectic structures}, {\it Ann. Phys.} (N. Y.) {\bf 111} (1978) 61.

\bibitem{Bayen2} F. Bayen, M. Flato, C. Fronsdal, A. Lichnerowicz, D. Sternheimer: {\it Deformation theory and quantization II. Physical applications}, {\it Ann. Phys.} (N. Y.) {\bf 110} (1978) 111.

\bibitem{Bertolami1} O. Bertolami, J.G. Rosa, C. Arag\~ao, P. Castorina, D. Zappal\`a: Noncommutative gravitational quantum well, {\it Phys. Rev. D} {\bf 72} (2005) 025010.

\bibitem{Bohm} A. Bohm, and M. Gadella, {\it Dirac kets, Gamow vectors and Gel'fand triplets: the rigged Hilbert space formulation of quantum mechanics. Lectures in mathematical physics at the University of Texas at Austin}, (Springer-Verlag, Berlin, 1989).

\bibitem{Bordemann} M. Bordemann, N. Neumaier, S. Waldmann: Homogeneous Fedosov star products on cotangent bundles I. Weyl and standard ordering with differential operator representation, {\it Comm. Math. Phys.} {\bf 198} (1998) 363.

\bibitem{Bracken} A. Bracken, G. Cassinelli, J. Wood: Quantum symmetries and the Weyl-Wigner product of group representations, arxiv: math-ph/0211001.

\bibitem{Werner1} T. Br\"ocker, R.F. Werner: Mixed states with positive Wigner functions, {\it J. Math. Phys.} {\bf 36} (1995) 62.

\bibitem{Cannas} A. Cannas da Silva, {\it Lectures on symplectic geometry, Lecture Notes in Mathematics}, (Springer, 2001).

\bibitem{Carroll} S.M. Carroll, J.A. Harvey, V.A. Kosteleck\'y, C.D. Lane, T. Okamoto: Noncommutative field theory and Lorentz violation, {\it Phys. Rev. Lett.} {\bf 87} (2001) 141601.

\bibitem{Connes} A. Connes: {Noncommutative geometry}, (Academic Press, 1994).

\bibitem{Demetrian} M. Demetrian, D. Kochan: Quantum mechanics on noncommutative plane, {\it Acta Phys. Slov.} {\bf 52} (2002) 1.

\bibitem{Dias1} N.C. Dias, J.N. Prata: Formal solutions of stargenvalue equations, {\it Ann. Phys.} (N. Y.) {\bf 311} (2004) 120.

\bibitem{Dias2} N.C. Dias, J.N. Prata: Admissible states in quantum phase space, {\it Ann. Phys.} (N. Y.) {\bf 313} (2004) 110.

\bibitem{Dias3} N.C. Dias, J.N. Prata: Generalized Weyl-Wigner map and Vey quantum mechanics, {\it J. Math. Phys.} {\bf 42} (2001) 5565.

\bibitem{Dias4} N.C. Dias, J.N. Prata: Time dependent transformations in deformation quantization, {\it J. Math. Phys.} {\bf 45} (2004) 887.

\bibitem{Dias5} N.C. Dias, J.N. Prata: Exact master equation for a noncommutative Brownian particle, {Ann. Phys.} (N.Y.) {\bf 324} (2009) 73.

\bibitem{Dias6} N.C. Dias, J.N. Prata: Narcowich-Wigner spectrum of a pure state, {Rep. Math. Phys.} {\bf 63} (2009) 43.

\bibitem{Douglas} M.R. Douglas, N.A. Nekrasov: Noncommutative field theory, {\it Rev. Mod. Phys.} {\bf 73} (2001) 977.

\bibitem{Dubin} D. Dubin, M. Hennings, T. Smith: {\it Mathematical aspects of Weyl quantization}, (World Scientific, Singapore, 2000).

\bibitem{Duval} C. Duval, P.A. Horvathy: Exotic galilean symmetry in the noncommutative plane and the Landau effect, {\it J. Phys. A} {\bf 34} (2001) 10097.

\bibitem{Ellinas} D. Ellinas, A.J. Bracken: Phase-space-region operators and the Wigner function: geometric constructions and tomography, {\it Phys. Rev. A} {\bf 78} (2008) 052106.

\bibitem{Fedosov1} B. Fedosov: A simple geometric construction of deformation quantization, {\it J. Diff. Geom.} {\bf 40} (1994) 213.

\bibitem{Fedosov2} B. Fedosov, {\it Deformation Quantization and Index Theory},
(Akademie Verlag, Berlin, 1996).

\bibitem{Folland} G.B. Folland, {\it Harmonic analysis in phase space}, (Princeton University Press, 1989).

\bibitem{Gamboa} J. Gamboa, M. Loewe, J.C. Rojas: Noncommutative quantum mechanics, {\it Phys. Rev. D} {\bf 64} (2001) 067901.

\bibitem{Gelfand} I.M. Gel'fand, G.E. Shilov, N.Y. Vilenkin: {\it Generalized Functions, Vols 1-5}, (Academic Press, New York, 1964-68).

\bibitem{Giedke} G. Giedke: {\it Quantum information and continuous variable systems, PhD Thesis}, (Innsbruck, 2001).

\bibitem{Gosson1} M. de Gosson, F. Luef: Quantum states and Hardy's formulation of the uncertainty principle: a symplectic approach, {\it Lett. Math. Phys.} {\bf 80} (2007) 69.

\bibitem{Gosson2} M. de Gosson: {\it Symplectic geometry and quantum mechanics}, (Birkh\"auser, 2006).

\bibitem{Groenewold} H. Groenewold: On the principles of elementary quantum mechanics, {\it Physica} {\bf 12} (1946) 405.

\bibitem{Grubb} G. Grubb: {\it Distributions and operators}, (Springer, 2009).

\bibitem{Hormander} L. H\"ormander: {\it The analysis of linear partial differential operators I}, (Springer-Verlag, 1983).

\bibitem{Horvathy} P. A. Horvathy: The noncommutative Landau problem, {\it Ann. Phys.} (N. Y.) {\bf 299} (2002) 128.

\bibitem{Hudson} R.L. Hudson: When is the Wigner quasi-probability density non-negative?, {Rep. Math. Phys.} {\bf 6} (1974) 249.

\bibitem{Kastler} D. Kastler: The $C^*$-algebras of a free boson field, {\it Commun. Math. Phys.} {\bf 1} (1965) 14.

\bibitem{Kontsevich} M. Kontsevich: Deformation quantization of Poisson manifolds, {\it Lett. Math. Phys.} {\bf 66} (2003) 157.

\bibitem{Littlejohn} R.G. Littlejohn: The semiclassical evolution of wave packets, {\it Phys. Rep.} {\bf 138} (1986) 193.

\bibitem{Loupias} G. Loupias, S. Miracle-Sole: $C^*$-algebres des systemes canoniques  Ann. Inst. H. Poincar\'e {\bf 6} (1967) 39.

\bibitem{Madore} J. Madore: {\it An introduction to noncommutative differential geometry and its physical applications, 2nd edition}, (Cambridge University Press, 2000).

\bibitem{Moyal} J. Moyal: Quantum mechanics as a statistical theory, {\it Proc. Camb. Phil. Soc.} {\bf 45} (1949) 99.

\bibitem{Pfaffian} T. Muir, W.H. Metzler, {\it A treatise on the theory of determinants}, (Courier Dover Publications, 2003).

\bibitem{Nair} V.P. Nair, A.P. Polychronakos: Quantum mechanics on the noncommutative plane and sphere, {\it Phys. Lett. B} {\bf 505} (2001) 267.

\bibitem{Narcowich1} F.J. Narcowich: Conditions for the convolution of two Wigner distributions to be itself a Wigner distribution, {\it J. Math. Phys.} {\bf 29} (1988) 2036.

\bibitem{Pool} J.C. Pool: Mathematical aspects of the Weyl correspondence, {\it J. Math. Phys.} {\bf 7} (1966) 66.

\bibitem{Roberts} J.E. Roberts: The Dirac bra and ket formalism, {\it J. Math. Phys.} {\bf 7} (1966) 1097.

\bibitem{Robertson} H.P. Robertson: The uncertainty principle, {\it Phys. Rev.} {\bf 34} (1929) 163.

\bibitem{Schrodinger} E. Schr\"odinger, Sitzungsber. Preuss. Akad. Wiss. 24 (1930) 296.

\bibitem{Segal} I.E. Segal: Transforms for operators and symplectic automorphisms over a locally compact abelian group, {\it Math. Scand.} {\bf 13} (1963) 31.

\bibitem{Seiberg} N. Seiberg, E. Witten: String theory and noncommutative geometry, {\it JHEP} {\bf 9909} (1999) 032.

\bibitem{Soto} F. Soto, P. Claverie: When is the Wigner function of multi-dimensional systems nonnegative?, {\it J. Math. Phys.} {\bf 24} (1983) 97.

\bibitem{Vey} J. Vey: D\'eformation du crochet de Poisson sur une vari\'et\'e symplectique, {\it Comment. Math. Helvet.} {\bf 50} (1975) 421.

\bibitem{Wigner} E. Wigner: On the quantum correction for thermodynamic equilibrium, {\it Phys. Rev.} {\bf 40} (1932) 749.

\bibitem{Wilde} M. Wilde, P. Lecomte: Existence of star-products and of formal deformations of the Poisson Lie algebra of arbitrary symplectic manifolds, {\it Lett. Math. Phys.} {\bf 7} (1983) 487.

\bibitem{Zemanian} A. Zemanian, {\it Distribution theory and transform analysis}, (Dover, 1987).


























































\end{thebibliography}
\end{document}